\documentclass[preprint,12pt]{elsarticle}

\usepackage{amssymb}
\usepackage{amsmath}
\usepackage{amsthm}
\usepackage{amsfonts}
\usepackage{tikz}
\usepackage{hyperref}
\usepackage{graphicx} 
\usepackage{lineno}

\usetikzlibrary{calc, shapes, fit, arrows.meta}

\newtheorem{theorem}{Theorem}

\newtheorem{lemma}{Lemma}

\journal{Applied Mathematical Modelling}

\begin{document}

    \begin{frontmatter}

        \title{A Game Theoretic Treatment of Contagion in Trade Networks}

        \author[label1,label3,label4]{John S. McAlister \corref{cor1}}
            \ead{jmcalis6@vols.utk.edu}
            \cortext[cor1]{Corresponding Author}
        \author[label5]{Jesse L. Brunner} 
        \author[label6]{Danielle J. Galvin} 
        \author[label1,label7,label3,label4]{Nina H. Fefferman}
        \affiliation[label1]{organization= {University of Tennessee - Knoxville, Department of Mathematics}}
        \affiliation[label7]{organization= {University of Tennessee - Knoxville, Department of Ecology and Evolutionary Biology}}
        \affiliation[label3]{organization= {National Institute for Modeling Biological Systems}}
        \affiliation[label4]{organization= {NSF Center for the Analysis and Prediction of Pandemic Expansion}}
        \affiliation[label5]{organization= {Washington State University, School of Biological Sciences}}
        \affiliation[label6]{organization= {University of Tennessee - Knoxville, School of Natural Resources}}

        \date{April 2025}

        \begin{abstract}
            Global trade of material goods involves the potential to create pathways for the spread of infectious pathogens. One trade sector in which this synergy is clearly critical is that of wildlife trade networks. This highly complex system involves important and understudied bidirectional coupling between the economic decision making of the stakeholders and the contagion dynamics on the emergent trade network. While each of these components are independently well studied, there is a meaningful gap in understanding the feedback dynamics that can arise between them. In the present study, we describe a general game theoretic model for trade networks of goods susceptible to contagion. The primary result relies on the acyclic nature of the trade network and shows that, through the course of trading with stochastic infections, the probability of infection converges to a directly computable fixed point. This allows us to compute best responses and thus identify equilibria in the game. We present ways to use this model to describe and evaluate trade networks in terms of global and individual risk of infection under a wide variety of structural or individual modifications to the trade network. In capturing the bidirectional coupling of the system, we provide critical insight into the global and individual drivers and consequences for risks of infection inherent in and arising from the global wildlife trade, and any economic trade network with associated contagion risks.  
        \end{abstract}


        \begin{highlights}
            \item Our novel game theoretic model captures feedback between economic and health related factors in the wildlife trade
            \item Acyclic wildlife trade networks have a stationary probability of infection
            \item Best responses can be computed in acyclic trade networks to find Nash equilibria
        \end{highlights}

    \end{frontmatter}

    \bibliographystyle{elsarticle-num} 


    \section{Introduction}
    Understanding infection and contagion in trade networks combines several important areas of study. The economic decision-making involved in the actions of trading and investing in measures to prevent infection, the contagion dynamics on a network, and the risk of spillover from the closed network into the environment, have each been studied on their own. All three elements separately are useful for informing bits and pieces of a system wherein many stakeholders are trading products which may transmit some infection. However, without considering each of the elements together, we cannot understand the bidirectional coupling between the health and safety practices and the economic decision making. Stakeholders make decisions about health and safety measures based on infection dynamics and perceived risk, the infection dynamics are changed by these decision which are also constrained by economic factors, and the structure of the network by which the players are connected may have a great impact on this interaction. Here we present a first model which integrates all of these components.

    Much of the work on the risks that result from wildlife trade has focused either on conservation or invasion biology \cite{meeks2024wildlife, patel2015quantitative, morton2021impacts}. While invasion from exotic wildlife being traded poses a large threat to ecosystems, it is not the only threat posed by such trade. Spillover of wildlife diseases poses a large threat as well but has thus far primarily been examined either statistically (looking for patterns in observed data) \cite{mitchell2024growth, Gippert2023global, Heinse2016Risk} or theoretically at an international scale, without regard for the economic incentives and behaviors of the trade network participants. Moreover, most of the focus has been on the spillover risk specifically to humans, e.g., \cite{biggs2023governance, Johnson2020Global,Magouras2020Emerging}, rather than also considering the impact to native wildlife or local agricultural populations. This is not to say we are alone; some studies have begun to consider these questions, examining the underlying economic incentive structures for human-environmental interactions, e.g., \cite{albers2020disease}. Some very elegant work has also explored the explicit economics of the interaction between animal trade and infection, but absent the framework of a trade network (e.g., \cite{horan2015managing}).

    In addition to the work done on invasion and spillover resulting from wildlife trade, there is a large body of work regarding the game theoretic elements of trade \cite{chu2024game, zhang2013farsighted, zhao2021complex}. Understanding economic behavior, even in complex networks is not the novel element of this work, rather we propose this model as a unique connection between the economic view of trade networks and biological view of spillover risk.

    Separate from both of these things,  there is a rich literature about the passage of a contagion or infection through networks (e.g., \cite{hill2010infectious, liu2014controlling, taylor2015topological, fefferman2007disease} and many others). Investigations of infection through networks have existed since the beginning of network science and has been used for applications ranging from understanding social strategies for viral marketing campaigns \cite{hinz2011social} to predicting spatio-temporal patterns in mortgage default behavior \cite{seiler2013strategic}. However, even when the economic and biological components have been studied together, to the best of our knowledge, there has been no work on the impact that the explicit network structure has on the coupled feedback between economic decision-making and infection risk in trade networks which is a crucial piece of designing better, more resilient trade networks. 

    While there have not been many predictive causal models proposed, there has been a rich body of work exploring how to understand contagion risks in animal trade networks. The preponderance of these have focused on the trade of agricultural and/or hunted animals (e.g., \cite{cristancho2021accounting, palisson2016role}), and nearly all have estimated risks based on ecological niche for relevant species and habitat characteristics necessary for exposure risks and statistical estimation of association data.

    Each of these types of models, on its own, is not sufficient to understand bidirectional coupling between health and safety practices and economic decision making. When stakeholders are making decisions about health and safety measures based on economic incentives which change based on the perceived risk, and the perceived risk changes in response to changes in health and safety decisions, we wee a new set of rich dynamics worth investigating. The key reason to understand this coupling is to advance our capabilities in predictive modeling towards making informed decisions about spillover risk from trade networks of products vulnerable to contagion. In the present study, we build a game theoretic model which is flexible and robust and that captures the relationship between investment in health and safety practices and interactions within the trade network. The model presented here is not predictive but provides a framework for understanding these kinds of systems with more accurate insight. The key reason that this model is usable comes from a result proving the existence of a long-time stable infection probability. This insight allows us to treat a typically stochastic system deterministically and compute best responses. Examining this system by studying best responses (and thus Nash equilibria) allows us to find points in the strategy space where players, making rational decisions to maximize their own payoff, are balancing the cost of health and safety practices with the benefit it gives them in the trade network.

    In section \ref{Model}, we describe the general model and the crucial hypotheses that allow it to be usable. In section \ref{StationaryProbability}, we present the result about the existence of the directly computable stationary probability of infection and describe the process to derive it. Equipped with this result we discuss the existence of best response functions in section \ref{EquilibriumResults} and provide an easy set of sufficient conditions for the existence of such a function. In section \ref{numericalExample}, we fill in the general functions of the model with explicit functional forms that give us a toy model to examine numerically. We discuss the results from the toy model and its implications for the application areas of this model in section \ref{discussion}.

    \section{Trade Network Contagion}\label{Model}
    Consider a group of $n$ players each with a strategy $x_i\in \mathbb{R}^s$ and with a contagion status $I_i\in \{0,1\}$. Let the collection of all $n$ strategies be called $X:=[x_1,x_2,...,x_n]\in \mathbb{R}^{s\times n}$. Furthermore, let the collection of all strategies except for that of player $i$ be called $X_{-i}\in \mathbb{R}^{s\times n-1}$ and let $I$ be the vector $[I_1,I_2,...,I_n]^T$.
    The game presently discussed considers these $n$ players, which interact with one another in an acyclic trade network, passing goods downstream through the trade network at rates which depend on the strategies of all players in the network.

    Our addition to the game is to consider how the presence of a contagion, which may infect and be transmitted by the goods being traded, changes the dynamics of the game. We consider only acyclic trade networks which means that contagion can only be passed in one direction, ``downstream," through the network, after it is introduced with probability $\epsilon_i$ from the environment to player $i$. This assumes a dynamic in which infections risks are associated with the transfer of traded animals themselves, rather than by mere contact between participants. To that end, we give the players an order so that if $i$ and $j$ interact and $i$ is contagious, $j$ may become contagious if and only if $i<j$. This is equivalent to giving the vertices in the directed graph representation of the trade network a topological ordering which necessarily exists because the graph is acyclic. The probability that a focal individual $i$ interacts with player $j<i$ (upstream) at a particular time is given by $a_{i,j}(X,I_j)$. The payoff for such an interaction for the focal individual is given by  $c_{i,j}(x_i,x_j)$. The probability that a focal individual $i$ interacts with a player $j>i$ (downstream) at a particular time is given by $a_{j,i}(X,I_i)$ and the payoff for that focal individual is given as $b_{i,j}(x_i,x_j)$.  Lastly, let $d_i(x_i,I_i)$ be the component of payoff which does not depend on interactions with other players in the network. Through all of this, we get the following payoff function.
    \begin{equation}\label{payoff}
        \pi_i(x_i,X_{-i},I)= d_i(x_i,I_i)+\sum_{j=1}^{i-1}c_{i,j}(x_i,x_j)a_{i,j}(X,I_j)+\sum_{j=i+1}^nb_{i,j}(x_i,x_j)a_{j,i}(X,I_i).
    \end{equation}

    We now express these as matrices and we get $A(X,I)$, a strictly lower triangular matrix with rows $\|\vec a\|_{L^1}\leq 1$; $B$, an upper triangular matrix; and $C$ a lower triangular matrix. Further, let $F$ be component wise multiplication of $C$ and $A^T$, and let $G$ be the component wise multiplication of $B$ and $A$. Let $d=[d]_{i=1}^n$ and we get
    \begin{equation*}
        \pi(X,I)=d(X,I)+(F(X,I)+G(X,I))\mathbf{1}
    \end{equation*}
    where $\mathbf{1}$ is the vector of all 1s, $I$ is the vector of contagion states, and $X$ is the strategy profile across all players (an $n\times m$ matrix). 

    Underlying this system is a system of contagion where each player, $x_i$, with a contagion state $I_i=0$ may become infected with some probability given by $\alpha_i(x_i)$, by interacting with an infected individual, or with probability $\beta_i(x_i)$ from an environmental source. Additionally,  each player $x_j$ with contagion state $I_j=1$ may recover with a probability given by $f(x_i)$ as shown in figure \ref{infectionstate}.
    \begin{figure}[h!]
        \centering
        \begin{tikzpicture}
            \node(a)[circle, draw, inner sep =5pt] at (0,0){$0$};
            \node(b)[circle, draw, inner sep =5pt] at (2,0){$1$};
    
            \node (start) at (-1.5,0){};
    
            \draw [line width = 1.5pt, -{>>}](start)--(a);
            \draw [line width = 1.5pt, -{>}](b)to[out=120, in= 60](a);
            \draw [line width = 1.5pt, -{>}](a)to[out=-60, in= -120](b);
            \draw [line width = 1pt, -{>}](a)--(b);
    
            \node(label1) at (1,1.3){$f(\sigma)$};
            \node(label2) at (1,-1.3){Transaction with $v$ where $I(v)=1$};
            \node(label3) at (1,0.3){$\varepsilon$};
        \end{tikzpicture}
        \caption{State diagram for the binary health state of any particular player in the game.}
        \label{infectionstate}
    \end{figure}

    We seek to solve this game by finding \textit{long time Nash equilibria} which are states when every player plays a strategy which, when played for a long time, gives an average payoff that cannot be improved upon unilaterally by that player \cite{Housman2023Game}. In other words, each player is playing a best response (in a long term average sense) to their co-players.

    \section{Stationary Probabilities of infection}\label{StationaryProbability}
    To seek long time Nash equilibria, we first must find the expected value of being infected at any time. Here we present a result which says that, regardless of the initial contagion state of the system, if contagion is introduced by environmental factors at a low amount and percolated through a network downstream in the fashion described above, the probability of being infected at any given time will converge to a stationary probability vector, so long as the graph is finite. The crucial element is to note that, because we can give the vertices an ordering so that the contagion is only passed from smaller vertices to larger vertices, this graph, as a digraph of contagion, is acyclic (it can be given a topological ordering) (Fig \ref{ExampleTradeNetwork}). 

    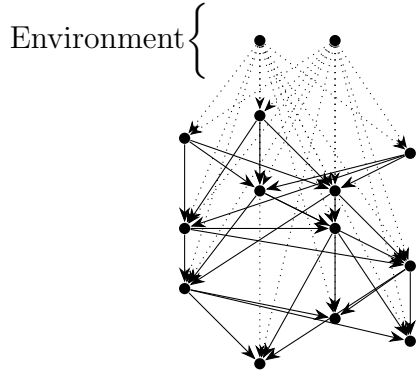
\begin{figure}[h!]
        \centering
        \begin{tikzpicture}
            \node(Env)[circle,fill,inner sep =1.5] at (1,1){};
            \node(Env2)[circle,fill,inner sep =1.5] at (2,1){};
            
		  \node(a)[circle, fill, inner sep =1.5pt] at (0,-0.3){};
            \node(b)[circle, fill, inner sep = 1.5pt] at(1,0){};
            \node(d)[circle, fill, inner sep = 1.5pt] at(3,-0.5){};

            \node(e)[circle, fill, inner sep =1.5pt] at (0,-1.5){};
            \node(f)[circle, fill, inner sep = 1.5pt] at(1,-1){};
            \node(h)[circle, fill, inner sep = 1.5pt] at(2,-1){};

            \node(i)[circle, fill, inner sep =1.5pt] at (0,-2.3){};
            \node(j)[circle, fill, inner sep = 1.5pt] at(2,-1.5){};
            \node(l)[circle, fill, inner sep = 1.5pt] at(3,-2){};

            \node(m)[circle, fill, inner sep =1.5pt] at (1,-3.3){};
            \node(n)[circle, fill, inner sep = 1.5pt] at(2,-2.7){};
            \node(p)[circle, fill, inner sep = 1.5pt] at(3,-3){};

            \draw[arrows = {-Stealth[scale=1.2]}](a)--(e);
            \draw[arrows = {-Stealth[scale=1.2]}](a)--(f);
            \draw[arrows = {-Stealth[scale=1.2]}](a)--(h);
            \draw[arrows = {-Stealth[scale=1.2]}](b)--(e);
            \draw[arrows = {-Stealth[scale=1.2]}](b)--(f);
            \draw[arrows = {-Stealth[scale=1.2]}](b)--(h);
            \draw[arrows = {-Stealth[scale=1.2]}](d)--(e);
            \draw[arrows = {-Stealth[scale=1.2]}](d)--(f);
            \draw[arrows = {-Stealth[scale=1.2]}](d)--(h);
            \draw[arrows = {-Stealth[scale=1.2]}](e)--(i);
            \draw[arrows = {-Stealth[scale=1.2]}](e)--(j);
            \draw[arrows = {-Stealth[scale=1.2]}](e)--(l);
            \draw[arrows = {-Stealth[scale=1.2]}](f)--(i);
            \draw[arrows = {-Stealth[scale=1.2]}](f)--(j);
            \draw[arrows = {-Stealth[scale=1.2]}](f)--(l);
            \draw[arrows = {-Stealth[scale=1.2]}](h)--(i);
            \draw[arrows = {-Stealth[scale=1.2]}](h)--(j);
            \draw[arrows = {-Stealth[scale=1.2]}](h)--(l);
            \draw[arrows = {-Stealth[scale=1.2]}](i)--(m);
            \draw[arrows = {-Stealth[scale=1.2]}](i)--(n);
            \draw[arrows = {-Stealth[scale=1.2]}](i)--(p);
            \draw[arrows = {-Stealth[scale=1.2]}](j)--(m);
            \draw[arrows = {-Stealth[scale=1.2]}](j)--(n);
            \draw[arrows = {-Stealth[scale=1.2]}](j)--(p);
            \draw[arrows = {-Stealth[scale=1.2]}](l)--(m);
            \draw[arrows = {-Stealth[scale=1.2]}](l)--(n);
            \draw[arrows = {-Stealth[scale=1.2]}](l)--(p);

            \draw[arrows = {-Stealth[scale=1]},dotted](Env)--(a);
            \draw[arrows = {-Stealth[scale=1 ]},dotted](Env)--(b);
            \draw[arrows = {-Stealth[scale=1 ]},dotted](Env)--(d);
            \draw[arrows = {-Stealth[scale=1 ]},dotted](Env)--(e);
            \draw[arrows = {-Stealth[scale=1 ]},dotted](Env)--(f);
            \draw[arrows = {-Stealth[scale=1 ]},dotted](Env)--(h);
            \draw[arrows = {-Stealth[scale=1 ]},dotted](Env)--(i);
            \draw[arrows = {-Stealth[scale=1 ]},dotted](Env)--(j);
            \draw[arrows = {-Stealth[scale=1 ]},dotted](Env)--(l);
            \draw[arrows = {-Stealth[scale=1 ]},dotted](Env)--(m);
            \draw[arrows = {-Stealth[scale=1 ]},dotted](Env)--(n);
            \draw[arrows = {-Stealth[scale=1 ]},dotted](Env)--(p);

            \draw[arrows = {-Stealth[scale=1]},dotted](Env2)--(a);
            \draw[arrows = {-Stealth[scale=1 ]},dotted](Env2)--(b);
            \draw[arrows = {-Stealth[scale=1 ]},dotted](Env2)--(d);
            \draw[arrows = {-Stealth[scale=1 ]},dotted](Env2)--(e);
            \draw[arrows = {-Stealth[scale=1 ]},dotted](Env2)--(f);
            \draw[arrows = {-Stealth[scale=1 ]},dotted](Env2)--(h);
            \draw[arrows = {-Stealth[scale=1 ]},dotted](Env2)--(i);
            \draw[arrows = {-Stealth[scale=1 ]},dotted](Env2)--(j);
            \draw[arrows = {-Stealth[scale=1 ]},dotted](Env2)--(l);
            \draw[arrows = {-Stealth[scale=1 ]},dotted](Env2)--(m);
            \draw[arrows = {-Stealth[scale=1 ]},dotted](Env2)--(n);
            \draw[arrows = {-Stealth[scale=1 ]},dotted](Env2)--(p);

            \node(ENVlabel) at (-1,1){Environment$\bigg\{$};;

        \end{tikzpicture}
        \caption{A connected acyclic digraph representing paths of transmission through the trade network. The environment is represented by some finite number of source terms each with possibly different different rates of infection. The weights of the digraph are represented by the weighted adjacency matrix $W$.}
        \label{ExampleTradeNetwork}
    \end{figure}

    We assume that there are $m$ environmental sources which poses a constant threat of contagion, and so we add these as vertices with indegree 0 to our digraph. Clearly adding such a vertex does not ruin the acyclic nature of the digraph. Call the new adjacency matrix $W\in [0,1]^{(n+m)\times (n+m)}$ where 
    
    $$W=
    \begin{bmatrix}
        \mathbf{O}&\mathbf{O}\\
        [\epsilon]&A
    \end{bmatrix}$$
        
    where $\mathbf{O}$ is the zero matrix, and $[\epsilon]$ is the matrix describing the weight of the edges between environmental sources and players. We can add any number of environmental contagion sources and our digraph will remain acyclic. Additionally, because it is always a source of contagion, it will always be considered infected. Let $\hat{I}\in \{0,1\}^{n+m}$ be the concatenation of the infection state of the environment and the contagion state of the players in the game. Because the environment is a constant threat of contagion we get, if $i\leq m$,
    \begin{equation*}\label{environment} 
        P(\hat{I}_i(t+1)=1)=P(\hat{I}_i(t)=1)=1.
    \end{equation*}

    At any time step, a player may become infected, may remain infected, or may recover. Let us first think of the likelihood of becoming infected. Vertices which had indegree 0 before the addition of the environmental vertices have 
    \begin{equation*}
        P(\hat{I}_i(t+1)=1|\hat{I}_i(t)=0)= \sum_{j=1}^m\beta_{i,j}(x_i)w_{i,j}P(\hat{I}_j(t)=1)=\sum_{j=1}^m\beta_{i,j}(x_i)w_{i,j}. 
    \end{equation*}
    where $\beta_{i,j}$ is a function describing how likely it is for player $i$ to take up the contagion from the environmental source $j$, given $i$'s current strategy. This may be different for every player. 

    Now the downstream effects of interaction complicate the equations further down the network. Infection can still happen randomly from the environment but it also can happen through interacting with an infected player. The probability of interacting with an infected individual can be understood from $a_{i,j}(X,1)$. This gives us the resulting probability of infection

    \begin{equation}\label{distrbutors}
        P(\hat{I}_i(t+1)=1|\hat{I}_i(t)=0)= \sum_{j=1}^m\beta_{i,j}(x_i)w_{i,j} +\sum_{j=m+1}^{i-1}\alpha(x_i)w_{i,j}(X,1)P(\hat{I}_j(t)=1)
    \end{equation}

    where $\alpha_i(x_i)$ is the probability that $i$ becomes infected given that an interaction with an infected individual takes place. Notice here that we are assuming that the likelihood of being infected by multiple interactions in a single time step is negligible. Also notice that we must assume that $i$'s interaction with $j$ does not depend on the infection status of any upstream player other than $j$. 

    Now, let us consider the probability of recovery. For every (non-environmental) player it is the same 
    \begin{equation*}\label{Recovery}
        P(\hat{I}_i(t+1)=0|\hat{I}_i(t)=1)=f_i(x_i)
    \end{equation*}
    If we make the same assumption about small enough time steps so that no stakeholder can become infected and recover in the same time step (and the other way around), this gives us that 
    \begin{equation*}\label{Recovery2}
        P(\hat{I}_i(t+1)=1|\hat{I}_i(t)=1)=(1-f_i(x_i)).
    \end{equation*}
    Again, note that the environmental nodes do not recover.

    Notice that now we can compute for $i=1$ to $n+m$
    \begin{equation}\label{nonlinear}
        \begin{split} 
            P(\hat{I}_i(t+1)=1)&= P(\hat{I}_i(t+1)=1|\hat{I}_i(t)=0)+\\
            &\quad P(\hat{I}_i(t+1)=1|\hat{I}_i(t)=1) P(\hat{I}_i(t)=1)-\\
            &\quad P(\hat{I}_i(t+1)=1|\hat{I}_i(t)=0) P(\hat{I}_i(t)=1)
        \end{split} 
    \end{equation}
    With this temporal dynamic, it will be helpful to define the following  
    \begin{equation*}\label{LinearT1}
        S:=
        \begin{bmatrix}
            \mathbf{O}&\mathbf{O}\\
            [\mathbf{Env(X)}]& diag(\mathbf{\alpha(X)})A(X,\mathbf{1})
        \end{bmatrix}\in \mathbb{R}^{n+m\times n+m}
    \end{equation*}
    \begin{equation*}\label{LinearT2}
        R:=
        \begin{bmatrix}
            I_m&\mathbf{0}&\mathbf{0}&...&\mathbf{0}\\
            \mathbf{0}^T&1-f(\sigma_1)&0&...&0\\
            \mathbf{0}^T&0&1-f(\sigma_2)&...&0\\
            \vdots&\vdots&\vdots&&\vdots\\
            \mathbf{0}^T&0&0&...&1-f(\sigma_n)
        \end{bmatrix}\in \mathbb{R}^{n+m\times n+m}
    \end{equation*}

    where $\mathbf{O}$ is the zero matrix, $\mathbf{0}$ is the zero vector, $[\mathbf{Env(X)}]$ is the component wise multiplication $[[\beta_{i,j}(x_i)]]*[\epsilon]$, and $diag(\alpha(X))A(X,1)$ is the matrix product of the matrix with $\alpha_i(x_i)$ on its diagonal with the matrix $[[a_{i,j}(X,1)]]$. Note that this product is strictly lower triangular. 

    $S$ can be thought of as a weighted adjacency matrix for the acyclic digraph with the environmental source, weighted by the likelihood of transmission.  Let $p_t = P(I_i(t)=1)$ and notice that the first term in \eqref{nonlinear} can be expressed as $Sp_t$, the second term can be expressed as $Rp_t$, and the final term can be expressed as $Sp_t*p_t$, where $(*)$ is component wise multiplication. This gives us the nonlinear difference equation  
    \begin{equation}\label{nonlinear2}
        p_{t+1}=(S+R)p_t-(Sp_t)*p_t.
    \end{equation} 
    Under the assumption that $\alpha(x_i)<1$ and $\beta(x_i)$ is sufficiently small, we can work towards some nice convergence results. 

    Let our transformation be called $Tp:=(S+R)p+(Sp)*p$. It is clear to see that, by construction, $T:\{\mathbf{1}\}\times[0,1]^n\rightarrow \{\mathbf{1}\}\times[0,1]^n$ where $\mathbf{1}$ is the vector of all $1$s representing the infection state of the environmental nodes. Finding asymptotic infection probabilities across all the vertices will allow us to compute best responses directly, with the assumption that the system runs for a sufficiently long time without collapse during the transient dynamics.

    Although the ``next generation" operator $T$ is not linear, it is lower triangular which allows us to directly compute its fixed point. Moreover, regardless of our initial choice of $p$, applying the next generation matrix repeatedly will necessarily result in a sequence which converges to the fixed point so long as the acyclic digraph is finite. Moreover, that limit is certainly the unique fixed point of $T$ and is exactly computable through forward substitution. In the theorem the indexing is changed for clarity instead of there being $m$ environmental nodes and $n$ stakeholders, there are $m$ environmental nodes and $n-m$ stakeholders for a total of $n$ nodes.  

    \begin{theorem}\label{StationaryProbabilityThm}
        Suppose $G$ is an acyclic weighted digraph of order $n$, where only $m$ vertices have indegree 0, with weighted adjacency matrix $S\in [0,1]^{n\times n}$ with each row satisfying $\|\vec{s}\|_{L^1}\leq 1$. Further suppose $R\in [0,1]^{n\times n}$, is a diagonal matrix with exactly the first $m$ entries equal to 1. Let the nonlinear operator $T:\{1\}^m\times[0,1]^{n-m}\rightarrow \{1\}^m\times[0,1]^{n-m}$ be defined as $Tp:= Sp+Rp-Sp*p$ (Where $*$ is component wise multiplication). Under these conditions, the sequence $(T^kp)_{n=0}^\infty$ converges to 
        \begin{equation}\label{WTNMainThmEq1}
            [p^\star]_j = \begin{cases}
                1&j\leq m\\
                \frac{\sum_{i=1}^{j-1}s_{i,j}p^*_i}{1-r_{j,j}+\sum_{i=1}^{j-1}s_{i,j}p^*_i}&j>m
            \end{cases}
        \end{equation}
        for any $p\in \{1\}^m\times [0,1]^{n-m}$.
    \end{theorem}
    \begin{proof}  
        The digraph is acyclic so give the vertices a topological ordering which ensures that the vertices with indegree 0 are labeled $1,2,..., m$, the weighted adjacency matrix $S$ is strictly lower triangular, and $R$ has $r_{i,i}=1\iff i\leq m$. Let $p\in\{1\}^m\times[0,1]^{n-m}$ so $[p]_i=1$ for $i\leq m$. Note the following two points: First, $\sum_{i=1}^ms_{i,j}p_i\leq 1$ for all $j$ so $|r_{j,j}-\sum_{i=1}^ms_{i,j}p_i|<1$ for all $j>m$.  Secondly, the convergence $[T^{(k)}p]_i\rightarrow p^*_i$ is trivial for $i\leq m$. 
            
        Observe that $[Tp]_{m+1}=\sum_{i=1}^{m}s_{i,m+1}+(r_{m+1,m+1}-\sum_{i=1}^{m}s_{i,m+1})p_{m+1}$. This is because $S$ strictly lower triangular. Moreover we see easily that applying the transformation again gives us the recurrence relation
        \begin{equation*}
            [T^{(k)}p]_{m+1} = \sum_{i=1}^ms_{i,m+1}+(r_{m+1,m+1}-\sum_{i=1}^m s_{i,m+1})[T^{(k-1)}p]_{m+1}
        \end{equation*}    
        This sequence has the general form 
        \begin{equation*}
            [T^{(k)}p]_{m+1} = p_{m+1}\left(r_{m+1,m+1}-\sum_{i=1}^ms_{i,m+1}\right)^k + \frac{1-(\sum_{i=1}^ms_{i,m+1})^k}{1-r_{m+1,m+1}+\sum_{i=1}^ms_{i,m+1}}
        \end{equation*}
        Thus, because $|r_{m+1,m+1}-\sum_{i=1}^m s_{i,m+1}|<1$,  for any $\varepsilon_1>0$, $\exists K_1$ such that for $k>K_1$, $|[T^{(k)}p]_{m+1}-p^*_{m+1}|<\varepsilon_1$, that is $[T^{(k)}p]_{m+1}\rightarrow p^*_{m+1}$.

        Now for the inductive step we will use the elementary analysis lemma (Appendix). Let $l>m+1$  and notice that 
        \begin{equation*}
            [T^{(k+1)}p]_l=\sum_{i=1}^{l-1}s_{i,l}[T^{(k)}p]_i+\left(r_{l,l}-\sum_{i=1}^{l-1}s_{i,l}[T^{(k)}p]_i\right)[T^{(k)}p]_l
        \end{equation*}
        Let $y_k:=[T^{(k+1)}p]_l$ and let $[T^{(k)}p]_i$ for $i=1,...,l-1$ together form the vector $\vec x_k$. Further let,  $a(\vec{x})=\sum_{i=1}^{l-1}s_{i,l}\vec{x}_i$ and $b(\vec x)= r_{l,l}-\sum_{i=1}^{l-1}s_{i,l}\vec{x}_i$. $a,b\in C^0(\mathbb{R}^{l-1},\mathbb{R})$. 
        Thus, if we rewrite the above equation as 
        \begin{equation*}
            y_k =a(\vec x_{k-1})+b(\vec x_{k-1})y_{k-1}
        \end{equation*}
        and note that $\lim \vec x_{k} = \vec x:= [p_i^*]_{i=1}^{l-1}$, the lemma says that $y_k\rightarrow \frac{a(x)}{1-b(x)}$ as $k\rightarrow \infty$. Translating back, this means that
        \begin{equation*}
            \lim _{k\rightarrow \infty} [T^{(k)}p]_l=\frac{\sum_{i=1}^{l-1}s_{i,l}p^*_i}{1-r_{l,l}+\sum_{i=1}^{l-1}s_{i,l}p^*_i}=p^*_l
        \end{equation*}
    
        This forms the basis for our inductive argument. Because $[T^{(k)}p]_i=1$ for all $k$ when $i\leq m$ and because $[T^{(k)}p]_{m+1}\rightarrow p^*_{m+1}$, the inductive step tells us that $[T^{(k)}p]_{m+2}\rightarrow p^*_{m+2}$. We may repeat this a \textit{finite} number of times to see that $[T^{(k)}p]_i\rightarrow p^*_i$ for all $i=1,...,n$ which we write as $\lim_{k\rightarrow \infty}T^{(k)}p=p^*$

        This works regardless of our initial choice for $p$ and so we have shown the desired result.       
    \end{proof}

    This result allows us to compute exact expected values for $\pi(X,I)$ and easily approximate (or, with great effort, compute exactly) how payoffs change with respect to the components of the strategy. In this way, we can find best responses. This provides us a direct mapping from the set of global parameters and vertex specific parameters to payoffs $w:\mathbb{R}^{n\cdot m}\rightarrow \mathbb{R}^n$
    \begin{equation}\label{convexPayoff}
        w(X)= g(X,P^*(X))
    \end{equation} 
    More specifically we can write 
    \begin{equation}\label{DeterministicPayoff}
        \begin{split}
            w_i(x_i,X_{-i})&=p^*_id_i(x_i,1)+ (1-p^*_i)d_i(x_i,0)\\
            &+\sum_{j=1}^{i-1}c_{i,j}(x_i,x_j)(p^*_ja_{i,j}(X,1)+(1-p^*_j)a_{i,j}(X,0))\\
            &+\sum_{j=i+1}^nb_{i,j}(x_i,x_j)(p^*_ia_{j,i}(X,1)+(1-p^*_i)a_{j,i}(X,0))
        \end{split}
    \end{equation}
    This means that we can carry out sensitivity analyses on the global parameters and the graphical structure and easily measure things like long time infection probabilities across the whole network.

    \section{Equilibrium Results}\label{EquilibriumResults}
    Equipped with the result from theorem \ref{StationaryProbabilityThm}, we can turn our attention to the analysis of this system. In order to find Nash equilibria, we must discuss best responses. In particular, we seek a function $\psi_i:\mathbb{R}^{s\cdot (n-1)}\rightarrow \mathbb{R}^s$ which, given a strategy profile $X_{-i}$ returns player $i$'s best response. If such a function exists, we can go further to find a function $\Psi(X)=[\psi_i(X_{-i})]_{i=1}^n$ and say that surely fixed points of $\Psi$ are Nash equilibria of the game (This is direct from the definition of Nash equilibrium).

    The particular form of $\Psi$ depends a great deal on the forms of the functions mentioned in sections \ref{Model} and \ref{StationaryProbability}.  For this reason, we do not attempt to find general results about $\Psi$ but we do present an argument for the existence of such a $\Psi$ under certain Hypotheses. 
        \subsection{Existence of a Best Response Function}
        We would like to discuss the existence of $x_i$ which maximizes $w_i(x_i,X_{-i})$. In the case that $X\in \Omega\subset \mathbb{R}^{x\times n}$ where $\Omega$ is compact, then continuity of $w$ is sufficient for the existence of at least one maximizer of $w_i(x_i,X_{-i})$ though a ``continuous on compact" argument.
    
        In order to show the existence of a maximizer in the case where strategies are unbounded, we need some restrictions on the behavior of the functions $b, c,$ and $d$, the functions which describe the upstream downstream and intrinsic payoff for a player. We will do this by considering the typical economic assumptions of \textit{increasing marginal costs} and \textit{decreasing marginal benefit}. Most simply, for our function $d$ we will require that $d_x(x_i,1)$ and $d_x(x_i,0)$ are both eventually decreasing in every component of $x_i$. Recall $d$ is the function which describes the component of payoff which is independent of other players. If we consider $x_i$ as different components of investment into different elements relating the game, then the laws of increasing marginal cost and decreasing marginal benefit will impose the assumption that eventually, increasing any investment will decrease payoff. 
    
        Recall that $c_{i,j}(x_i,x_j)$ governs the payoff to player $i$ from upstream interaction (i.e. purchasing). We can assume that $c_{i,j}$ is everywhere negative. Likewise recall that $b_{i,j}$ governs the payoff to player $i$ from downstream interaction (i.e. selling). We assume that $b_{i,j}$ is everywhere positive. A simple way we can ensure the existence of a best response is by assuming that both $b_{i,j}(x_i,x_j)$ and $c_{i,j}(x_i,x_j)$ are bounded in the first argument for all $i,j$. The assumption can be justified through the finiteness of resources available to stakeholders. Surely each player has a limit to how much they can buy from players upstream and how much they can sell to players downstream. These two assumptions together are the basis for Hypotheses 2 and 3 in the following theorem, forming a loose set of reasonable sufficient conditions for the existence of a best response function.
    
        Under these assumptions, with a strategy space equipped with tie-breaking order, there then exists a function $\psi_i:\mathbb{R}^{s\times(n-1)}\rightarrow \mathbb{R}^s$. Such that $\psi_i(X_{-i})$ is $i$'s best response to $X_{-i}$. Let $\Psi:\mathbb{R}^{s\times n}\rightarrow \mathbb{R}^{s\times n}$ be the concatenation of each of these functions so that $\Psi(X)=[\psi_i(X_{-i})]_{i=1}^n$. This is written formally in proposition \ref{BestResponseExistence}. It is immediate from the definition of Nash equilibrium that a fixed point of $\Psi$ is a Nash Equilibrium of the game. For use in modeling, finding fixed points of $\Psi$ is crucial and can be done numerically. 
    
        \begin{theorem} \label{BestResponseExistence}
            Let $\Omega\subseteq \mathbb{R}^{s\times n}$ be a strategy space for the game defined by the payoff equation \eqref{payoff}, equipped with a tie breaking ordering, $\prec_\Omega$. Under the following hypotheses
            \begin{enumerate}
                \item[H1)] $a_{i,j},b_{i,j},c_{i,j}, \beta_{i,j}$, $d_i$, $f_i,$ and $ \alpha_i$, as defined above, are all continuous in $x_i$ for all $i,j=1,2,...,n$
                \item[H2)] There exits $R$ such that $|x_i|>R\implies d_i(x_i,I_i)\leq M_0-\epsilon|x_i|$ for some $M$ and some $\epsilon$. 
                \item[H3)] $|c_{i,j}|$ and $|b_{i,j}|$ are bounded by $M$
            \end{enumerate}the function $\Psi:\Omega\rightarrow \Omega$ which satisfies the condition that $\Psi(X)$ is the concatenation of every players' best response to $X$ is well defined.  Moreover, if $\Omega$ is compact, then (H2) and (H3) are unnecessary. 
        \end{theorem}
    
        \begin{proof}
            let $\Omega = \Omega_{i}\times \Omega_{-i}$ where for any player $i$, $x_i\in \Omega_{i}$ and $X_{-i}\in \Omega_{-i}$. Payoff for each player $i$ is determined by $w_i(x_i,X_{-i})$ as in \eqref{DeterministicPayoff}
            For any particular $X_{-i}$ we seek to show that that the maximum of $w_i(x_i,X_{-i})$ over $x_i\in\Omega_{i}$ is attained.
    
            When $\Omega\subset\subset \mathbb{R}^{s\times n}$ then we need only show that $w_i$ is continuous. From (H1) it is obvious that $\pi_i(x_i,X_{-i},1)$ is the sum of continuous functions and so it itself continuous. Moreover, $p^\star_i(x_i,X_{-i})$ for any player $i$ is given as \eqref{WTNMainThmEq1}. Again, it is clear that from (H1) that $s_{i,j}$ is continuous for all $i,j=m+1,m+2,...,m+n$ (corresponding to all the players). If $f$ is continuous then so is $r_{i,i}$ and, as noted in Theorem \ref{StationaryProbabilityThm} the denominator of \eqref{WTNMainThmEq1} is always strictly positive. Therefore the quotient of these two continuous functions is continuous and thus so is $p^\star_i$. With the continuity having been determined, it is direct from the compactness of $\Omega$ that there is an $x^\star_i\in \Omega_{i}$ such that $w_i(x^\star_i,X_{-i})=\max_{\Omega_{i}}w_i(\cdot, x_{-i})$. The fact that the argmax of $w_i(\cdot, X_{-i})$ is nonempty means that the we certainly have a best response by way of the tie-breaking order$\prec_\Omega$.  Clearly a tie-breaking order is certain to exist so long as the strategy space has finite dimension.
    
            In the case that $\Omega$ is not compact, we put some decay estimates on $w_i$ so that we know the maximum is attained. Notice that $w_i$ is a convex combination of terms from $\pi_i(x_i,X_{-i},1)$ and $\pi_i(x_i,X_{-i},0)$ and so, obviously,

            $$\min_{I\in \{0,1\}^n}\pi_i(x_i,X_{-i},I)\leq w_i(x_i,X_{-i})\leq \max_{I\in \{0,1\}^n}\pi_i(x_i,X_{-i},I).$$ 
            
            By (H3) it is easy to see that $\pi_i(x_i,X_{-i},I)\leq d_i(x_i,I_i)+Mn$ for any $I\in \{0,1\}^n$. Let $m:=\min_{I\in \{0,1\}^n}\pi_i(\mathbf{0},X_{-i},I).$ Let $R_1=\max\{R, \frac{-m +M_0+Mn}{\epsilon}\}$. When $|x_i|> R_1$, then $d_i(x_i,I)<m-Mn$ and thus $\pi(x_i,X_{-i},I)<m$, for any $I$. This implies that 
            $$\max_{x_i\in\Omega\setminus B_{R_0}(0)}\pi_i(x_i,X_{-i},I)<m\leq \max_{x_i\in\overline{B_{R_0}(0)}}\pi_i(x_i,X_{-i},I)$$
            for any I. This means that outside of $B_{R_0}(0)$, $w_i(x_i,X_{-i}) < m $ and, because we now have that $m\leq \sup_{\Omega} w_i(x_i,X_{-i})$,
            we can say that $\sup_\Omega w_i(x_i,X_{-i})=\max_{\overline{B_{R_0}(0)}}w_i(x_i,X_{-i}),$ which surely exists by the same continuity on compact argument. Now that we know that the maximum of $w(x_i,X_{-i})$ is attained on $\Omega$, then surely we know that, given the tie breaking order $\prec_{\Omega_i}$, every $X_{-i}$ corresponds to one best response $x_i$. 
    
            Thus we have shown that for any $X_{-i}$, whether $\Omega$ is compact or not, there exists a best response. Let $\psi_i:\mathbb{R}^{s\times(n-1)}\rightarrow \mathbb{R}^s$ be the function which maps $X_{-i}$ to its best response. This is well defined (although we have no hope of describing its behavior or regularity). Easily concatenate these functions for all $i$ to get $\Psi:\mathbb{R}^{s\times n}\rightarrow \mathbb{R}^{s\times n}$ where $\Psi(X)=[\psi_i(X_{-i})]_{i=1}^n$. This function is well defined, and maps a strategy profile $X$ to the strategy profile that where each player is playing a best response to $X$.  
        \end{proof}
    
        Because we have only shown that the best response function is well defined, but have done nothing to describe its behavior or regularity, we can do very little to determine the existence of a fixed point much less the stability of such a fixed point. However, it does mean that the process of myopic best response is well defined and that there is, perhaps, a process by which fixed points may be found numerically.  In $\Omega$ we may solve the problem $X-\Psi(X)=0$. Without regularity of $\Psi$, or indeed the a priori knowledge that a solution exists, we are not guaranteed  through this solution concept, but we can use myopic best response to examine the evolution of strategies in time.  

    \section{A Numerical Example through Myopic Best response}\label{numericalExample}
    Here we present a toy example to demonstrate that this method of considering trade network games can be used to reveal surprising results about how the structure of a trade network may impact the strategic equilibria. Consider a game with $N$ players organized in an acyclic digraph with a strictly lower triangular adjacency matrix $W$ where each row sums to 1. This is a trade network where each edge describes an ``upstream" transaction (i.e. a purchase).  As before suppose each player as a one dimensional strategy $x_i\in [0,1]$ and an infection state $I_i\in \{0,1\}$. In this toy example, "strategy" is mildly modeled after ``investment into health and safety measures" but is far too simple to capture that fully. 

    Suppose that the likelihood of transaction $a_{ij}(X,I_j)=w_{i,j}(1-x_iI_j)$ so as player $i$'s strategy increases, their likelihood to interact with an infected player decreases. Recall that total payoff can be separated into three components: The benefit to player $i$ from interacting upstream (purchasing), the benefit to player $i$ from interacting downstream (selling), and the intrinsic benefit to player $i$. Consider the game where each of these components for player $i$ interacting with a player $j$  are defined as follows
    \begin{equation*}
        \begin{split}
            \text{downstream:}\quad\quad & b(x_i,x_j)= (1-x_i)\\
            \text{upstream:}\quad\quad & c(x_i,x_j)=(x_j-1)\\
            \text{intrinsic:} \quad\quad & d(x_i,I_i)=x_i(x_i-1)(1-I_i) 
        \end{split}
    \end{equation*}
    Notice that the upstream and downstream components satisfy the relation $c(x_i,x_j)=-b(x_j,x_i)$. With each of these components, we can compute the total payoff as
    \begin{equation}\label{EXpayoff}
        \begin{split}
            \pi(x_i,x_{-i},I)=x_i&(x_i-1)(1-I_i)+\sum_{j=1}^{i-1}(x_j-1)w_{i,j}(1-x_iI_j)\\
            &+\sum_{j=i+1}^N(1-x_j)w_{j,i}(1-x_jI_i)
        \end{split}
    \end{equation}

    We also make the decision that $\alpha_{i,j}(x_i)=\beta_{i,j}(x_i)=(1-x_i)$ and that the probability of recovery is exactly $f(x_i)=x_i$ With these assumptions, we can write down our process for computing $p^*$ as described in Theorem \ref{StationaryProbabilityThm}.

    \begin{equation*}
        p^*_i=\begin{cases}
            1&i=0\\
            \frac{\epsilon(i-x_i)+\sum_{j=1}^{i-1}w_{i,j}(1-x_j)^2p^*_j}{x_i+\epsilon(i-x_i)+\sum_{j=1}^{i-1}w_{i,j}(1-x_j)^2p^*_j}& i>0
        \end{cases}
    \end{equation*}

    From this, we can write the long-run payoff function $w$ as 
    \begin{equation*}
        \begin{split}
            w(x_i,x_{-i}) &= x_i(x_i-1)(1-p_i^*)\\
            &+\sum_{j=1}^{i-1}(x_j-1)w_{i,j}(1-x_ip^*_j)\\
            &+\sum_{j=i+1}^n(1-x_j)w_{j,i}(1-x_jp^*_i)
        \end{split}
    \end{equation*}
    Notice that, because of the linearity of this system, we can write $w_i(x_i,x_{-i})=\pi_i(x_i,x_{-i},p^*_i)$. This is not generally true when the functional forms are non-linear. Because we now have a map $w:[0,1]^n\rightarrow \mathbb{R}^n$ which takes the strategy profile $X\mapsto w(X)$ a vector of fitnesses, we can use proposition \ref{BestResponseExistence} to show that the map $\Psi:[0,1]^n\rightarrow [0,1]^n$ which takes a strategy profile $X$ and returns a best response, $\Psi(X)$, is well defined given a certain tie-breaking order, $\prec_\Omega$. Although we cannot investigate this function analytically in this example, we can investigate it numerically. In this way we can approximate fixed points of the map $\Psi$ by repeatedly optimizing $w_i(x_i,X_{-i})$ with respect to the first variable for each $i$. This process is exactly the myopic best response process from evolutionary game theory and will terminate in a Nash equilibrium (if it terminates). Because there is no guarantee that $\Psi$ has a unique fixed point, we will run the simulation multiple times from many initial strategy profiles as a way of taking a low resolution image of the basin of stability for the solutions we find. 

        \subsection{Proof of concept on symmetric and asymmetric networks} 
        As a sanity check we show that network structure has measurable effects on the strategic equilibria of the game. Consider two networks which have the same sets of players but slightly different patterns of interaction. 
        \begin{figure}[h!]
            \centering
            \begin{tikzpicture}
                
    		  \node(a)[circle, fill, inner sep =1.5pt] at (0,0){};
                \node(b)[circle, fill, inner sep = 1.5pt] at(1,0){};
                \node(c)[circle, fill, inner sep = 1.5pt] at(0,-1){};
                \node(d)[circle, fill, inner sep =1.5pt] at (1,-1){};
                \node(e)[circle, fill, inner sep = 1.5pt] at(-1,-2){};
                \node(f)[circle, fill, inner sep = 1.5pt] at(0,-2){};
    
                \node(g)[circle, fill, inner sep =1.5pt] at (1,-2){};
                \node(h)[circle, fill, inner sep = 1.5pt] at(2,-2){};

                \draw[arrows = {-Stealth[scale=1.2]}](a)--(d);
                \draw[arrows = {-Stealth[scale=1.2]}](a)--(c);
                \draw[arrows = {-Stealth[scale=1.2]}](b)--(d);
                \draw[arrows = {-Stealth[scale=1.2]}](b)--(c);
                \draw[arrows = {-Stealth[scale=1.2]}](c)--(e);
                \draw[arrows = {-Stealth[scale=1.2]}](c)--(f);
                \draw[arrows = {-Stealth[scale=1.2]}](c)--(g);

                \draw[arrows = {-Stealth[scale=1.2]}](d)--(f);
                \draw[arrows = {-Stealth[scale=1.2]}](d)--(g);
                \draw[arrows = {-Stealth[scale=1.2]}](d)--(h);
    
                 \node(s1) at (-0.5,-1){s1};
                 \node(s2) at (1.5,-1){s2};
                 \node(a2) at (3.5, -1){a1};
                 \node(a2) at (5.5, -1){a2};
    
                \node(i)[circle, fill, inner sep =1.5pt] at (4,0){};
                \node(j)[circle, fill, inner sep = 1.5pt] at(5,0){};
                \node(k)[circle, fill, inner sep = 1.5pt] at(4,-1){};
                \node(l)[circle, fill, inner sep =1.5pt] at (5,-1){};
                \node(m)[circle, fill, inner sep = 1.5pt] at(3,-2){};
                \node(n)[circle, fill, inner sep = 1.5pt] at(4,-2){};
                \node(o)[circle, fill, inner sep =1.5pt] at (5,-2){};
                \node(p)[circle, fill, inner sep = 1.5pt] at(6,-2){};
          
                \draw[arrows = {-Stealth[scale=1.2]}](i)--(k);
                \draw[arrows = {-Stealth[scale=1.2]}](i)--(l);
                \draw[arrows = {-Stealth[scale=1.2]}](j)--(k);
                \draw[arrows = {-Stealth[scale=1.2]}](j)--(l);
                \draw[arrows = {-Stealth[scale=1.2]}](k)--(m);
                \draw[arrows = {-Stealth[scale=1.2]}](k)--(n);
                \draw[arrows = {-Stealth[scale=1.2]}](k)--(o);
                \draw[arrows = {-Stealth[scale=1.2]}](k)--(p);
                \draw[arrows = {-Stealth[scale=1.2]}](l)--(p);
            \end{tikzpicture}
            \caption{Two trade networks with 8 players. The arrows point in the direction that goods are passed through the system (which is different from the way that the adjacency matrix $W$ encodes the information). On the \textbf{left} is the symmetric case and on the \textbf{right} is the asymmetric case. The two nodes of interest in the symmetric case are labeled $s1$ and $s2$ and in the asymmetric case those same players are labeled $a1$ and $a2$.}
            \label{SymmetryProofOfConcept}
        \end{figure}
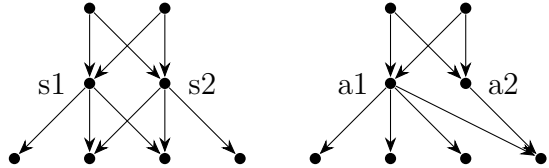
        We will use the example of two nearly identical three layer trade networks which differ only in the edges between levels two and three (Fig \ref{SymmetryProofOfConcept}). In the symmetric case there are two producers, two distributors (s1 and s2) which each buy from both producers evenly, and four consumers. One consumer buys only from s1, one buys only from s2 and the other two but from both s1 and s2 evenly. In the asymmetric case, again there are two producers, two distributors (a1 and a2) which buy from the producers evenly and four consumers. However, in this network three consumers buy only from a1 and the fourth buys evenly from a1 and a2. In the asymmetric network, a2 has only one potential customer.  
    
        \begin{table}[]
            \centering
            \begin{tabular}{c|cc|cc|}
                 &\multicolumn{2}{|c|}{Symmetric}&\multicolumn{2}{|c|}{Asymmetric}  \\
                 &s1&s2&a1&a2\\
                 \hline
                Equilibrium Strategy&0.4377&0.4377&0.3924&0.5195\\
                Infection Probability&0.1944&0.1944&0.2314&0.1391\\
                \hline
                Na\"ive Risk& \multicolumn{2}{|c|}{0.7276}&\multicolumn{2}{|c|}{0.7298}\\
                Weighted Risk&\multicolumn{2}{|c|}{0.4614}&\multicolumn{2}{|c|}{0.5365}
            \end{tabular}
            \caption{The results of the equilibrium search through myopic best response on the symmetric and asymmetric network reported for the focal individuals s1, s2, a1, and a2. In the symmetric case, the equilibrium strategies are equal while in the asymmetric case, the distributor with the greater ``market share" has a lower strategy (which again can be considered low investment in health and safety measures for the purposes of application) than the one with the smaller market share. The probabilities of infection at equilibrium are also reported for each player. In the symmetric game both players have the same probability of infection, in the asymmetric game the distributor with the greater market share has a higher probability of infection. Risk measurements for each network are reported. Regardless of which measure of risk is chosen, the asymmetric network poses a greater risk.}
            \label{SymmetricTable}
        \end{table}
    
        The results, summarized in table \ref{SymmetricTable} for the distributors, show that, as expected, the structure of the network has a reasonable impact on the results of the game. It is a helpful sanity check to confirm that symmetric conditions result in symmetric solutions. It is also important to note that the numerics seem to support the monostability of the system as the same equilibrium solutions are found in both the asymmetric and symmetric cases regardless of the initial data (plus or minus some numerical error which was observed as high as $5\times 10^{-4}$, perhaps because of the rather inelegant method of optimization used here). 
    
        The main result is that the symmetric system resulted in symmetric results and that in the asymmetric system, a1, the sole distributor to the majority of the consumers, had a lower strategy at equilibrium than a2. It was also observed that the long term probability of infection for a1 was higher than that of a2. This is discussed further in section \ref{discussion}. 
    
        In global terms we may also want to consider the total probability of infection. The na\"ive way is to consider the probability that none of the stakeholders are infected which is easily calculable. We take the ``Na\"ive risk" to be $1-\prod_{i=1}^n(1-p^*_i)$, which is the probability that at least one individual is infected. In this example the symmetric case has a na\"ive risk of 0.7276 and the asymmetric case has a na\"ive total probability of 0.7298. This approach does not take into account the idea that an infection at the top of the trade network pose a greater threat to the network as a whole.
    
        In order to discuss a better measure of global risk we introduce the idea of a weighted risk measurement. Recall that $A(X,I_j)$ is the weighted adjacency matrix describing the probability of transacting upstream given a strategy profile $X$ and an infection profile $I$. Suppose that individual $i$ is infected. The expected number of  additional infections $i$ may cause after a single transaction are given by the sum of the i$^th$ column of the matrix $S$ where $s_{i,j}=\alpha(x_i)a_{i,j}=w_{i,j}(1-x_i)^2$. Likewise the expected number of additional infections $i$ may cause \textit{indirectly} from a sequence of two transactions is given by the sum of the $i^{th}$ column of $S^2$. We can repeat this process for increasing path lengths until eventually $S^k=\mathbf{O}$ (which is certain to happen because $S$ is strictly lower triangular and thus nilpotent of degree $\leq n$.) By adding each of these expected number of additional infections we get an estimation of $\mathcal{R}_0^i$ for each player $i$.
        
        $$\mathcal{R}_0^i=\sum_{k=1}^n \mathbf{1}^TS^n\hat{e}_i$$ 
        
        This quantity represents the total number of expected new infections $i$ may cause should $i$ itself become infected. This gives us our weighted risk measurement which we call $\mathcal{R}_g$ which we use as a measure of spillover risk from the network. Although more sophisticated measures of spillover risk exist, for our purposes this weighted risk measurement gives an appropriate estimation of the number of new infections appearing in the long run equilibrium state. If we assume each infection behaves as an independent opportunity to yield a spillover event, then this equilibrium state will be approximately proportional to the number of spillover cases.   
        
        $$\mathcal{R}_g=\sum_{i=1}^n\mathcal{R}_0^ip^*_i.$$
    
        Applying this measurement to our toy example we see that the symmetric case has a weighted risk of $0.4614$ and the asymmetric case has a weighted risk of $0.5365$. Assured that our model works and can be reasonably assessed, we can proceed to discuss other questions about the toy model.
        
        \subsection{Monopoly effects in the toy model}
        \begin{figure}[h!]
            \centering
            \begin{tikzpicture}
                \node(label1) at (-1.5,-0.5){\textbf{c}};

    		  \node(a)[circle, fill, inner sep =1.5pt] at (0,0){};
                \node(b)[circle, fill, inner sep = 1.5pt] at(1,0){};
                
                \node(c)[circle, fill, inner sep = 1.5pt] at(-0.25,-1){};
                \node(d)[circle, fill, inner sep =1.5pt] at (0.5,-1){};
                \node(e)[circle, fill, inner sep = 1.5pt] at(1.25,-1){};

                \node(f)[circle, fill, inner sep = 1.5pt] at(-1.75,-2){};
                \node(g)[circle, fill, inner sep =1.5pt] at (-1.25,-2){};
                \node(h)[circle, fill, inner sep = 1.5pt] at(-0.75,-2){};
                \node(i)[circle, fill, inner sep = 1.5pt] at(-0.25,-2){};
                \node(j)[circle, fill, inner sep = 1.5pt] at(0.25,-2){};
                \node(k)[circle, fill, inner sep = 1.5pt] at(0.75,-2){};
                \node(l)[circle, fill, inner sep = 1.5pt] at(1.25,-2){};
                \node(m)[circle, fill, inner sep = 1.5pt] at(1.75,-2){};
                \node(n)[circle, fill, inner sep = 1.5pt] at(2.25,-2){};
                \node(o)[circle, fill, inner sep = 1.5pt] at(2.75,-2){};

                \draw[arrows = {-Stealth[scale=1.2]}](a)--(d);
                \draw[arrows = {-Stealth[scale=1.2]}](a)--(c);
                \draw[arrows = {-Stealth[scale=1.2]}](a)--(e);

                \draw[arrows = {-Stealth[scale=1.2]}](b)--(c);
                \draw[arrows = {-Stealth[scale=1.2]}](b)--(d);
                \draw[arrows = {-Stealth[scale=1.2]}](b)--(e);
                
                \draw[arrows = {-Stealth[scale=1.2]}](c)--(f);
                \draw[arrows = {-Stealth[scale=1.2]}](c)--(g);
                \draw[arrows = {-Stealth[scale=1.2]}](c)--(h);
                \draw[arrows = {-Stealth[scale=1.2]}](c)--(i);
                \draw[arrows = {-Stealth[scale=1.2]}](c)--(j);
                \draw[arrows = {-Stealth[scale=1.2]}](c)--(k);
                \draw[arrows = {-Stealth[scale=1.2]}](c)--(l);
                \draw[arrows = {-Stealth[scale=1.2]}](c)--(m);
                \draw[arrows = {-Stealth[scale=1.2]}](c)--(n);
                \draw[arrows = {-Stealth[scale=1.2]}](c)--(o);
    
                \draw[arrows = {-Stealth[scale=1.2]}](d)--(f);
                \draw[arrows = {-Stealth[scale=1.2]}](d)--(g);
                \draw[arrows = {-Stealth[scale=1.2]}](d)--(h);
                \draw[arrows = {-Stealth[scale=1.2]}](d)--(i);
                \draw[arrows = {-Stealth[scale=1.2]}](d)--(j);
                \draw[arrows = {-Stealth[scale=1.2]}](d)--(k);
                \draw[arrows = {-Stealth[scale=1.2]}](d)--(l);
                \draw[arrows = {-Stealth[scale=1.2]}](d)--(m);
                \draw[arrows = {-Stealth[scale=1.2]}](d)--(n);
                \draw[arrows = {-Stealth[scale=1.2]}](d)--(o);
    
                \draw[arrows = {-Stealth[scale=1.2]}](e)--(f);
                \draw[arrows = {-Stealth[scale=1.2]}](e)--(g);
                \draw[arrows = {-Stealth[scale=1.2]}](e)--(h);
                \draw[arrows = {-Stealth[scale=1.2]}](e)--(i);
                \draw[arrows = {-Stealth[scale=1.2]}](e)--(j);
                \draw[arrows = {-Stealth[scale=1.2]}](e)--(k);
                \draw[arrows = {-Stealth[scale=1.2]}](e)--(l);
                \draw[arrows = {-Stealth[scale=1.2]}](e)--(m);
                \draw[arrows = {-Stealth[scale=1.2]}](e)--(n);
                \draw[arrows = {-Stealth[scale=1.2]}](e)--(o);

    
            \end{tikzpicture}
            \hspace{1.5cm}
            \begin{tikzpicture}
    
                \node(label1) at (2.5,-0.5){\textbf{mm}};
                \node(a)[circle, fill, inner sep =1.5pt] at (0,0){};
                \node(b)[circle, fill, inner sep = 1.5pt] at(1,0){};
                
                \node(c)[circle, fill, inner sep = 1.5pt] at(-0.25,-1){};
                \node(d)[circle, fill, inner sep =1.5pt] at (0.5,-1){};
                \node(e)[circle, fill, inner sep = 1.5pt] at(1.25,-1){};

                \node(f)[circle, fill, inner sep = 1.5pt] at(-1.75,-2){};
                \node(g)[circle, fill, inner sep =1.5pt] at (-1.25,-2){};
                \node(h)[circle, fill, inner sep = 1.5pt] at(-0.75,-2){};
                \node(i)[circle, fill, inner sep = 1.5pt] at(-0.25,-2){};
                \node(j)[circle, fill, inner sep = 1.5pt] at(0.25,-2){};
                \node(k)[circle, fill, inner sep = 1.5pt] at(0.75,-2){};
                \node(l)[circle, fill, inner sep = 1.5pt] at(1.25,-2){};
                \node(m)[circle, fill, inner sep = 1.5pt] at(1.75,-2){};
                \node(n)[circle, fill, inner sep = 1.5pt] at(2.25,-2){};
                \node(o)[circle, fill, inner sep = 1.5pt] at(2.75,-2){};

                \draw[arrows = {-Stealth[scale=1.2]}](a)--(d);
                \draw[arrows = {-Stealth[scale=1.2]}](a)--(c);
                \draw[arrows = {-Stealth[scale=1.2]}](a)--(e);

                \draw[arrows = {-Stealth[scale=1.2]}](b)--(c);
                \draw[arrows = {-Stealth[scale=1.2]}](b)--(d);
                \draw[arrows = {-Stealth[scale=1.2]}](b)--(e);
                
                \draw[arrows = {-Stealth[scale=1.2]}](c)--(f);
                \draw[arrows = {-Stealth[scale=1.2]}](c)--(g);
    
                \draw[arrows = {-Stealth[scale=1.2]}](d)--(f);
                \draw[arrows = {-Stealth[scale=1.2]}](d)--(g);
                \draw[arrows = {-Stealth[scale=1.2]}](d)--(h);
                \draw[arrows = {-Stealth[scale=1.2]}](d)--(i);
                \draw[arrows = {-Stealth[scale=1.2]}](d)--(j);
                \draw[arrows = {-Stealth[scale=1.2]}](d)--(k);
                \draw[arrows = {-Stealth[scale=1.2]}](d)--(l);
                \draw[arrows = {-Stealth[scale=1.2]}](d)--(m);
                \draw[arrows = {-Stealth[scale=1.2]}](d)--(n);
                \draw[arrows = {-Stealth[scale=1.2]}](d)--(o);
    
                \draw[arrows = {-Stealth[scale=1.2]}](e)--(n);
                \draw[arrows = {-Stealth[scale=1.2]}](e)--(o);

    
            \end{tikzpicture}\\
            \vspace{0.5cm}
            \begin{tikzpicture}
    
                \node(label1) at (-1.5,-0.5){\textbf{tm}};
                \node(a)[circle, fill, inner sep =1.5pt] at (0,0){};
                \node(b)[circle, fill, inner sep = 1.5pt] at(1,0){};
                
                \node(c)[circle, fill, inner sep = 1.5pt] at(-0.25,-1){};
                \node(d)[circle, fill, inner sep =1.5pt] at (0.5,-1){};
                \node(e)[circle, fill, inner sep = 1.5pt] at(1.25,-1){};

                \node(f)[circle, fill, inner sep = 1.5pt] at(-1.75,-2){};
                \node(g)[circle, fill, inner sep =1.5pt] at (-1.25,-2){};
                \node(h)[circle, fill, inner sep = 1.5pt] at(-0.75,-2){};
                \node(i)[circle, fill, inner sep = 1.5pt] at(-0.25,-2){};
                \node(j)[circle, fill, inner sep = 1.5pt] at(0.25,-2){};
                \node(k)[circle, fill, inner sep = 1.5pt] at(0.75,-2){};
                \node(l)[circle, fill, inner sep = 1.5pt] at(1.25,-2){};
                \node(m)[circle, fill, inner sep = 1.5pt] at(1.75,-2){};
                \node(n)[circle, fill, inner sep = 1.5pt] at(2.25,-2){};
                \node(o)[circle, fill, inner sep = 1.5pt] at(2.75,-2){};

                \draw[arrows = {-Stealth[scale=1.2]}](a)--(d);
                \draw[arrows = {-Stealth[scale=1.2]}](a)--(c);
                \draw[arrows = {-Stealth[scale=1.2]}](a)--(e);

                \draw[arrows = {-Stealth[scale=1.2]}](b)--(c);
                \draw[arrows = {-Stealth[scale=1.2]}](b)--(d);
                \draw[arrows = {-Stealth[scale=1.2]}](b)--(e);
    
                \draw[arrows = {-Stealth[scale=1.2]}](d)--(f);
                \draw[arrows = {-Stealth[scale=1.2]}](d)--(g);
                \draw[arrows = {-Stealth[scale=1.2]}](d)--(h);
                \draw[arrows = {-Stealth[scale=1.2]}](d)--(i);
                \draw[arrows = {-Stealth[scale=1.2]}](d)--(j);
                \draw[arrows = {-Stealth[scale=1.2]}](d)--(k);
                \draw[arrows = {-Stealth[scale=1.2]}](d)--(l);
                \draw[arrows = {-Stealth[scale=1.2]}](d)--(m);
                \draw[arrows = {-Stealth[scale=1.2]}](d)--(n);
                \draw[arrows = {-Stealth[scale=1.2]}](d)--(o);


            \end{tikzpicture}
            \caption{Three trade networks with 15 players. The arrows point in the direction that goods are passed through the system (which is different from the way that the adjacency matrix $W$ encodes the information). On the \textbf{top left} is the competitive case, on the \textbf{top right} is the mild monopoly case, and on the \textbf{bottom} is the total monopoly case. We will be interested in the ``distributors" (middle row in each network) which are called ``c1,c2, and c3"  from left to right in the competitive case, ``mm1, mm2, and mm3" from left to right in the mild monopoly case, and ``tm1, tm2, and tm3" from left to right in the total monopoly case.}
            \label{Monopoly Model}
        \end{figure}
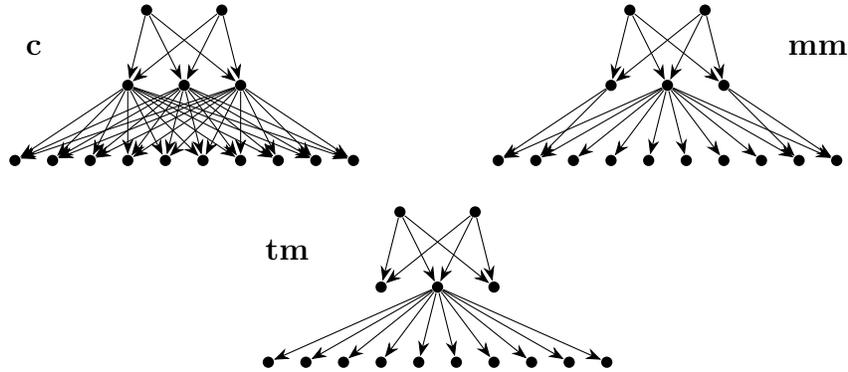

        To consider how market share plays a roll in infection risk for a wildlife trade network, we consider three networks with 15 players (Fig. \ref{Monopoly Model}). The first is the competitive network where there are two producers and three distributors (c1, c2, and c3) who each buy evenly from the two distributors. There are also 10 consumers who each buy evenly from each of the three distributors. This means that each of the distributor has exactly one third of the market share.  In the mild monopoly case there are again two producers and three distributors (mm1, mm2, and mm3) who each buy evenly from the producers. There are ten consumers. Two of them buy evenly from mm1 and mm2, two of them buy evenly from mm2 and mm3, and the remaining consumers buy exclusively from mm2. This means that mm1 and mm3 each have 10\% of the market share and mm2 has 80\% of the market share. Lastly, in the total monopoly case there are two producers and three distributors (tm1, tm2, and tm3) who buy evenly from each producer. The difference is that all of the ten consumers buy only from tm2. Clearly this means that tm2 has 100\% of the market share and tm1 and tm3 act as terminal consumers in the game.  

        \begin{table}[]
            \hskip-2cm
            \begin{tabular}{c|ccc|ccc|ccc|}
                &\multicolumn{3}{|c|}{Competative}&\multicolumn{3}{|c|}{Mild Monopoly}&\multicolumn{3}  {|c|}{Total Monopoly} \\
                &c1&c2&c3&mm1&mm2&mm3&tm1&tm2&tm3\\
                \hline
                Equilibrium Strategy&0.3451&0.3451&0.3451&0.4334&0.3082&0.4334&0.5506&0.3027&0.5506\\
                Infection Probability&0.3734&0.3734&0.3734&0.2628&0.4220&0.2628&0.1549&0.4031&0.1549\\
                \hline
                Na\"ive Risk& \multicolumn{3}{|c|}{0.9830}&\multicolumn{3}{|c|}{0.9771}&\multicolumn{3}{|c|}{0.9691}\\
                Weighted Risk&\multicolumn{3}{|c|}{3.2393}&\multicolumn{3}{|c|}{3.2525}& \multicolumn{3}{|c|}{3.3292}
            \end{tabular}
            \caption{A comparison of the strategies, infection probabilities, and risk measurements from the networks shown in figure \ref{Monopoly Model}}
            \label{Monopoly Table}
        \end{table}

        As the monopolistic quality of the trade network increases we see that the weighted risk increases as well although the na\"ive risk decreases. This can be explained by the fact that in a monopoly, the majority of the pathways of infection must pass through a single individual, the monopolist. This may mean that the distributors with less market share can choose strategies which result in a lower probability of infection so the na\"ive risk decreases. However, the weighted risk captures the concentration of infective pathways through the monopolist and so weights their probability of infection higher. There results are summarized in table \ref{Monopoly Table}. 
    
        In this game, it is also the case that having more consumers allows a distributor to take on a lower strategy which sometimes results in a higher probability of infection. For example compare c2, mm2, and tm2, as the market share increases the strategy they choose (0.3451, 0.3082, 0.3072 respectively). Interestingly this does not translate directly to a higher probability of infection. In the case of the total monopoly, tm2 is actually less likely to become infected than mm2 because of the upstream effects of the monopoly on the producers. 

        In a real trade network, not all stakeholders can be neatly sorted into different categories like producer, distributor, or consumer. Many stakeholders take on more than one role. This makes the results slightly harder to discuss because we can not take neat steps up and down the trade network as we can here. However, the model presented here in no way depends on the separation of stakeholders in to these categories. The only restrictive assumption is that the flow of goods is acyclic in the network. 
    
        \subsection{Upstream cascades from parameter adjustments}
        Another important note about the model is that adjustment to parameters downstream can have real effects on the equilibrium strategies upstream. To demonstrate this we change the weight on the intrinsic benefit term for the consumers in the networks shown in figure \ref{SymmetryProofOfConcept}. Ranging this weight from 0 (in the case where no consumer receives any payoff other than the cost of transacting with distributors) to 1.5 (in the case where the intrinsic benefit absolutely outweighs the cost of transacting with distributors.) We can see qualitative shifts in the behavior of distributors and even producers. 
        \begin{figure}
            \centering
            \includegraphics[width=0.8\linewidth]{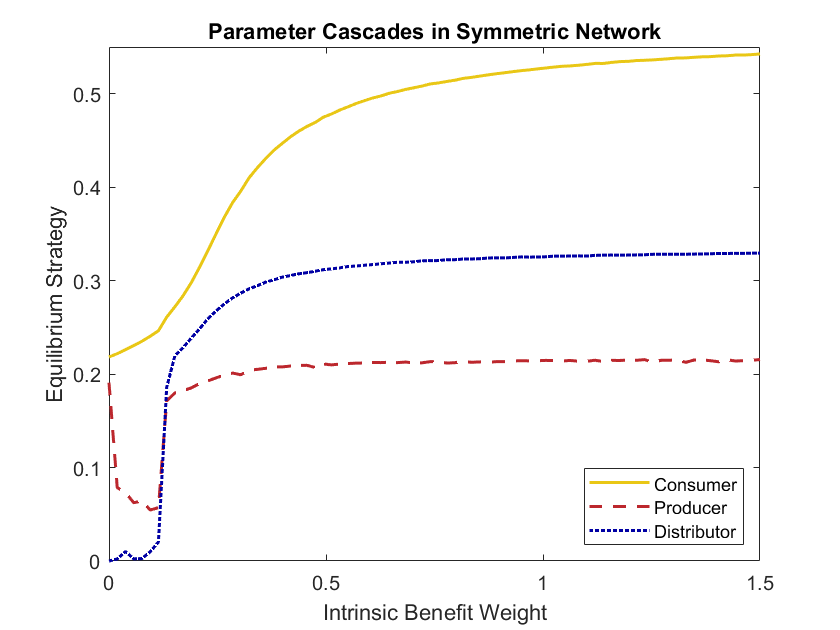}
            \includegraphics[width = 0.8\linewidth]{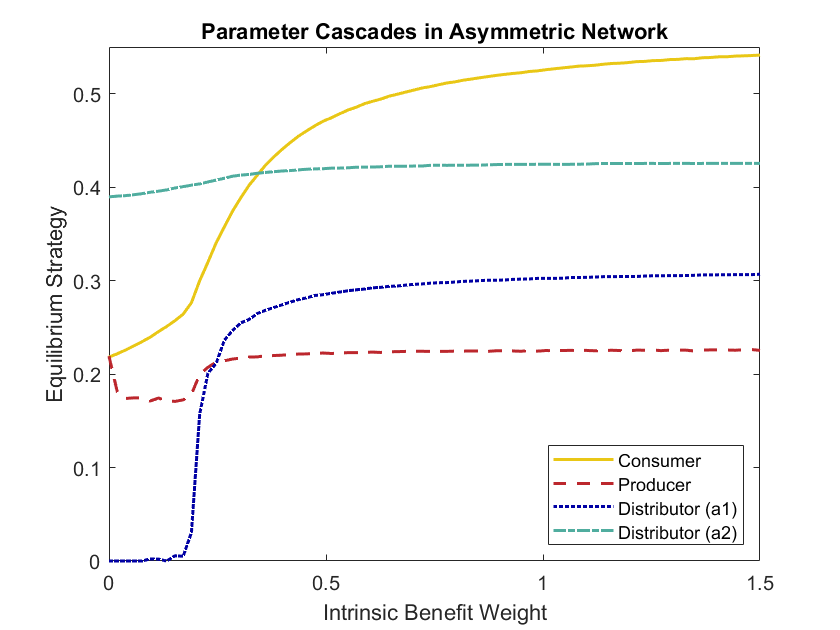}
            \caption{Changing the intrinsic benefit weight uniformly for all consumers from 0 to 1.5 shows a qualitative change in strategy for both symmetric (\textbf{top}) and asymmetric (\textbf{bottom}) examples. The consumer equilibrium strategy changes continuously with weight, the distributor strategy changes sharply at a tipping point between 0.2 and 0.3 and the producer strategy changes  abruptly at the same tipping point although this change appears to be less pronounced in the asymmetric network because of the insulating effects of the asymmetry discussed previously.}
            \label{Cascades Figure}
        \end{figure}

        The behavior shown in figure \ref{Cascades Figure} is hard to describe, especially because the toy model is not complex enough to capture real trade network-like behavior, but the main point is clear. As intrinsic Benefit is weighted less for consumers, their equilibrium strategy decreases. This is to be expected because, in the extreme limiting case without an intrinsic benefit, there is no cost of infection for consumers. The result which is at first surprising (then entirely natural once considered further) is the tipping points demonstrated by the very rapid change in equilibrium strategy taken on by the distributors and the producers. When intrinsic benefit weight shrinks below a particular threshold, distributors are no longer incentivized to invest in health and safety measures and so their equilibrium strategy drops to zero. Around the same point, the equilibrium strategy of the producers also decreases. This change for the produces is more pronounced in the symmetric case than in the asymmetric case because in the asymmetric case, the distributor which provides only to a single individual (a2) must maintain a higher equilibrium strategy and therefore provides an incentive for the producers to maintain a higher equilibrium strategy. We do not, at this moment, have an explanation for the sudden uptick in producer strategy when intrinsic benefit approaches 0. 

        This is an important feature of the model because it demonstrates its usefulness in considering interventions for disease. Throughout a trade network, different stakeholders have to play by different rules so controls at one layer of a network may be more feasible than others. By demonstrating that the entire network is sensitive to changes to parameters at a single level, we demonstrate the usefulness of this model for that purpose. This idea is discussed further in section \ref{discussion}.
    
        \subsection{Defectors in the toy model}
        Finally, we want to see the role that Defectors have in this model. In this case, we do not mean that each stakeholder has the choice to cooperate or defect, rather we mean that a single stakeholder may choose to take on a particular strategy, regardless of the payoff, because of some behavior or belief not captured by the model. For instance, choosing not to pasteurize milk because of an scientifically inaccurate political belief. 

        In the first case, we take the symmetric network in figure \ref{SymmetryProofOfConcept} and consider a consumer, who interacts with both distributors evenly. We measure average equilibrium behavior in the case that the consumer plays rationally, they defect to the strategy $x=0$ or they defect to $x=1$. The results are summarized in table \ref{ConsumerDefectTable}.

        In the second case we take the same network and consider a defecting distributor in both the extremes. The results are summarized in table \ref{DistributorDefectTable}. Lastly we discuss the same situation when a producer is the one to defect. These results are summarized in table \ref{ProducerDefectTable}.    

        \begin{table}[h!]
            \hskip-2cm
            \begin{tabular}{c|ccc|ccc|ccc|}
                &\multicolumn{3}{|c|}{Rational}&\multicolumn{3}{|c|}{Defect to 0}&\multicolumn{3}{|c|}{Defect to 1} \\
                &c&d&p&c&d&p&c&d&p\\
                \hline
                Equilibrium Strategy&0.5666&0.4374&0.3505&0.5668&0.4227&0.3556&0.5664&0.4444&0.3470\\
                Infection Probability&0.1236&0.1947&0.1564&0.1261&0.2052&0.1540&0.1228&0.1907&0.1595\\
                \hline
                Na\"ive Risk& \multicolumn{3}{|c|}{0.7277}&\multicolumn{3}{|c|}{1}&\multicolumn{3}{|c|}{0.6879}\\
                Weighted Risk&\multicolumn{3}{|c|}{0.4618}&\multicolumn{3}{|c|}{0.4867}& \multicolumn{3}{|c|}{0.4555}
            \end{tabular}
            \caption{A comparison of the strategies, infection probabilities, and risk measurements resulting from a defecting consumer in the symmetric trade network in figure    \           \ref{SymmetryProofOfConcept}. Equilibrium Strategy and Infection Probability are listed for a rational consumer (c), a distributor (d) and a producer (p).}
            \label{ConsumerDefectTable}
        \end{table}
        \begin{table}[h!]
            \hskip-2cm
            \begin{tabular}{c|ccc|ccc|ccc|}
                &\multicolumn{3}{|c|}{Rational}&\multicolumn{3}{|c|}{Defect to 0}&\multicolumn{3}{|c|}{Defect to 1} \\
                &c&d&p&c&d&p&c&d&p\\
                \hline
                Equilibrium Strategy&0.5666&0.4374&0.3505&0.5650&0.4351&0.3280&0.5695&0.4244&0.3865\\
                Infection Probability&0.1236&0.1947&0.1564&0.2178&0.2031&0.1470&0.0970&0.1958&0.1371\\
                \hline
                Na\"ive Risk& \multicolumn{3}{|c|}{0.7277}&\multicolumn{3}{|c|}{1}&\multicolumn{3}{|c|}{0.6018}\\
                Weighted Risk&\multicolumn{3}{|c|}{0.4622}&\multicolumn{3}{|c|}{2.4849}& \multicolumn{3}{|c|}{0.2671}
            \end{tabular}
            \caption{A comparison of the strategies, infection probabilities, and risk measurements resulting from a defecting distributor in the symmetric trade network in figure \ref{SymmetryProofOfConcept}. Equilibrium Strategy and Infection Probability are listed for a consumer (c), a rational distributor (d) and a producer (p).}
            \label{DistributorDefectTable}
        \end{table}
        \begin{table}[h!]
            \hskip-2cm
            \begin{tabular}{c|ccc|ccc|ccc|}
                &\multicolumn{3}{|c|}{Rational}&\multicolumn{3}{|c|}{Defect to 0}&\multicolumn{3}{|c|}{Defect to 1} \\
                &c&d&p&c&d&p&c&d&p\\
                \hline
                Equilibrium Strategy&0.5666&0.4374&0.3505&0.5666&0.4400&0.3521&0.5671&0.4274&0.3579\\
                Infection Probability&0.1236&0.1947&0.1564&0.1619&0.3503&0.14559&0.1147&0.1606&0.1525\\
                \hline
                Na\"ive Risk& \multicolumn{3}{|c|}{0.7277}&\multicolumn{3}{|c|}{1}&\multicolumn{3}{|c|}{0.6338}\\
                Weighted Risk&\multicolumn{3}{|c|}{0.4622}&\multicolumn{3}{|c|}{2.1790}& \multicolumn{3}{|c|}{0.314}
            \end{tabular}
            \caption{A comparison of the strategies, infection probabilities, and risk measurements resulting from a defecting producer in the symmetric trade network in figure    \ref{SymmetryProofOfConcept}. Equilibrium Strategy and Infection Probability are listed for a consumer (c), a distributor (d) and a rational producer (p).}
            \label{ProducerDefectTable}
        \end{table}

        The main result that we see from each situation is that a single defector can do very little to change the equilibrium strategy of players up- or downstream but they can have a great effect on the downstream infection probability. For this reason, we see that single defectors that are not terminal consumers can change the weighted risk of a network to an extreme degree. The na\"ive risk is nonsensical in this comparison because, when a defector defects to the strategy $x=0$, they will surely become infected and so the na\"ive risk will always be 1 in this case. 

    \section{Discussion and Applications}\label{discussion}
    Many individual aspects of strategy and payoff in trade networks with contagion have been studied independently. Building on the diversity of studies that have considered some subsets of factors relating to infection risks in animal trade networks \cite{biggs2023governance, meeks2024wildlife, morton2021impacts} and the economic factors that motivate human management/trade decisions about them \cite{zhang2013farsighted, mitchell2024growth, meeks2024wildlife, horan2015managing}, we have here proposed a model that allows for consideration of the confluence of all of these factors together.  In this work, we have addressed a gap in understanding the combined effects of each of these individual factors. While earlier related works have focused on estimating statistical likelihoods for spillover risks based on analysis and projection of observational data (e.g., \cite{albers2020disease}), or on characterizing specific trade network topologies (e.g., \cite{patel2015quantitative}), or on calculating the economic incentives of the overall system rather than of individuals in a network (e.g., \cite{horan2015managing}), we instead employ game theory on networks methods.  This allows us to consider the structure of the trade network as an independent variable and observe the resulting infection risk across the network, which we use as an estimate for both the risk of downstream contamination inside the network and spillover to na\"ive populations outside of the network.  

    Here we combine three crucial elements: (1) Economic Decision making in response to contagion, (2) Contagion dynamics in a network, and (3) Assessments of risk from that network. Our model captures the bidirectional coupling of the disease processes and the economic decision making in a way that is robust and flexible enough to be used to describe the qualitative likely behaviors and outcomes of contagion in trade networks. By using a game theoretic framework we capture the economic decision making of the stakeholders, including their responses to infection information. The disease process is captured very simply but the key insight is that, for any acyclic trade network, a stochastic infection process can be summarized with long-time expected infection probability. This expected infection probability gives us a way to calculate best responses and thus find equilibria. While this model does not have the specificity to make exact quantitative predictions, it does give ecologists, wildlife disease epidemiologists, and economists a way to describe and interrogate qualitative differences among proposed methods for intervention and/or control (e.g., via either standards for hygienic practice or economic incentive) across different network topologies. 

    Our model focuses on the scale of individual decision-makers within the trade network (as opposed to nations or sectors of the market). Each player's choices regarding trade partner connection and infection management practices are therefore emergent resulting properties of system, ultimately influenced by consumer pressures on economic and epidemiological behaviors. In this way, our model allows exploration into how intervention policies may propagate backwards through trade networks to increase healthy and protective behaviors. For any individual trade network and infection risk, some of the earlier work could be used to parameterize the explicit strategies and payoffs of our game, however our primary goal is not to make specific, quantitative recommendations for any one trade network, but rather to explore the emergent properties of all such systems and whether/how they may be influenced.
            
    The key feature of this model is that it captures the bidirectional coupling of individual choice in the optimization problem with respect to the global setting (i.e. the entire network) and the risk of infection that emerges from the global setting. Individuals make choices with risk of infection in mind, risk of infection changes as a result of these choices, and so the resulting equilibrium is not a product of either one individually but rather a product of their interactions. This means that, in addition to controls on contagion entering and being passed through the system, controls on the payoff structure, which determine individual behavior, are also potent to change risk, even on a global scale. 

    From the numerical example, we see that changing the payoff structure for subset of players in the model (as in example 5.3) can have a potent effect upstream. This means that our model is able to tease out the impact of a single stakeholders actions. In large systems like trade networks, it is easy to believe that individual actions do not impact overall outcomes measurably but this method of modeling provides us a way of describing the consequences of individual actions. In the case of defection, we see that in the toy model one player's choice to defect does little to change the strategies of the others but does greatly change the spillover risk, especially when the defector is far upstream. Although defection is not rational, and a stakeholder acting in their own best interest would not defect, it may be the case that through disinformation or some other means, stakeholders may be led to defect thinking it is in their own best interest.  This model may lend us to believe that having stakeholders who are well informed about the conditions of the network in which they trade can protect the interests of everyone in the network as well as reduce spillover risk outside the network. In addition, when considering the example of parameter cascades we see that changing incentive structures can have large impacts on the network as a whole. This is of note because, in the application area, it means that should external controls of payoff structures (e.g. subsidies for health and safety measures) be implemented in such a trade network, uniform enforcement in required in order for the controls to have the desired effect.

    Although our toy model does little to capture the specific conditions of a real life trade network, it does help us tease apart the mechanisms of upstream decision cascades. When an individual downstream has a change to their payoff function, their best response changes accordingly. When players downstream care more about health practices they force their upstream neighbors to invest more heavily in such practices, not for health's sake but rather driven by profit maximization. The effect is more noticeable when considered in the opposite direction. When consumers care little for health and safety practices, those decisions are readily passed upstream to distributors and producers. Distributors are free to maximize their profits by not investing at all in the health of their stock, so long as the consumers do not care. This may demonstrate the important insight that the upper limit for investment in health and safety upstream is difficult to increase from a downstream position but the floor is almost entirely dependent of the strategies of the downstream stakeholders. Strategies of distributors are limited above by the strategies of the producers upstream, those stakeholders which set the prices at which the distributors must operate. However, the lower limit is set by the stakeholders downstream and their considerations for the health of the products being distributed. 

    This insight, interestingly, does not entirely extend to the case of defection from a downstream position. When a consumer defects we see that the strategies of the distributors and the producers change very little. This is due to the particular form of the toy model. When a consumer defects to zero, the stakeholders upstream of that consumer get a payoff of exactly $w_{j,i}$ regardless of infection status, this means that in the optimization process, this terms in a constant and does not have an impact of the location of the optimum. Likewise, when a consumer defects to one, the stakeholders upstream get no payoff from transacting with the defector and so of course it will not impact the optimum. In either case, it is the same as if the defector was just removed from the network. Additionally, because consumers do little to impact the $R_0$ of the network, this defection has almost no measurable impact on the system as a whole. A single defection is dangerous only when it occurs upstream because such upstream effects can lead to increased infection throughout the entire network. It may also be the case that with sufficiently many consumers defecting, the distributors, beginning to be considered terminal stakeholders in the network, may have a qualitative shift in their behavior. This is why we note that a \textit{single} defection is only dangerous when it occurs upstream.

    This toy model, of course, does little to tell us about any particular real-world trade network, but it does demonstrate a promising new way to conceptualize and model these systems. In particular, it gives a direct way to make qualitative predictions about spill over risk and intrinsic risk given different control scenarios. Both individual controls and network wide controls can easily be tested through this model and both global and local measures of risk can easily be observed. The model is highly flexible and, provided a sensible tie-breaking ordering is supplied, a best response function can almost always be proven to exist. The computation required to do the optimization \cite{McAlisterSpilloverModel} is not computationally expensive and so, even with more complicated, nonlinear, functional forms, this model can be used numerically without issue to draw conclusions about multiple measures of risk in and around trade networks. 
            
    Of course, were we to attempt to parameterize our model to accurately and precisely reflect any individual real-world wildlife trade network, we could quantify the sensitivity of the relative outcomes for the emergent network structure of trade partnerships/volume and health protective behaviors. However, obtaining these data is itself so complicated as to be prohibitive during these first efforts. Recent work to understand consumers and trade participants beliefs about health practices have begun the costly and painstaking survey work needed to understand economic- and values-based incentives at work \cite{cavasos2023understanding, cavasos2023attitudes}. Estimating the probability of introduction of infection into a supply chain from the diversity of sources from which animals are provided varies greatly not only by pathogen of concern, but by also by a vast diversity of extrinsic environmental and ecological factors \cite{smith2017summarizing, can2019dealing}. Estimating the impact of hygienic practices faces the challenge not only of careful lab study, but then of understanding how recommended protocols may be enacted by untrained individuals \cite{berg2020compliance}. Each of these elements are individually complicated and costly to characterize well, and we therefore do not focus on any analysis of carefully parameterized scenarios - rather, we focus on the system-wide causes and consequences of frame shifts in expected behaviors and outcomes. 

    We have demonstrated a new way of capturing the bidirectional coupling of economic decision making and infection dynamics in trade networks of products susceptible to contagion. With the key insight that, as long as the trade network is acyclic, the long time probability of infection is stable and directly computable, we are able to find best responses with a very general set of payoff functions in this framework. With computable best responses and infection probabilities, the model we propose allows us to consider equilibria in trade network systems and measure the associated risk. This means this flexible and robust model is a potent tool to help us answer questions about controls in networks and their effects on infection/spillover risk.

    \section{Funding Statement}
    This work was supported by NSF DEB \#2207922 with additional support from NSF DBI \#2412115.
    \section{Data Statement}
    All the data in this manuscript was synthetically generated by the model using the code found in the repository \url{https://github.com/feffermanlab/JSM_2024_WildlifeTradeNetworks}
    \appendix 
        \section{A Necessary Lemma}
        \begin{lemma}[An elementary analysis lemma]
            Suppose $\vec x_k$ is a real sequence which converges to $\vec x$. Let $a,b:\mathbb{R}^n\rightarrow \mathbb{R}$ which are both continuous at $\vec x$, with $|b(\vec x)|<1$. Then the sequence given by the recursive form $y_k = a(\vec x_{k})+b(\vec x_{k})y_{k-1}$ has a limit 
            \begin{equation}
                \lim y_k = \frac{a(\vec x)}{1-b(\vec x)}
            \end{equation}
        \end{lemma} so long as $a$ and $b$ are such that $y_k >0$ for all $k$.
    
        \begin{proof}
            The proof is broken into three cases. The first is the case where $b(x)>0$. For any $\delta$ there is a $K_0$ such that $|b(x_k)-b(x)|<\delta $  and $|a(x_k)-a(x)|<\delta$ when $k>K_0$. Select $\delta$ so that $0<b(x)- \delta<b(x)+\delta <1$ and so that it satisfies a condition which will come later in the proof. By selecting a $\delta$ we get a requisite $K_0$.
            It is no loss of generality to reindex so that $(x_0,y_0)=(x_K,y_K)$.
    
            Throughout the proof I will use $a^+:=a(x)+\delta, a^-:=a(x)-\delta, b^+:=b(x)+\delta.$ and $b^-:=b(x)-\delta$
    
            This case is quite simple, we can bound each term of the sequence below by adding the lowest possible $a(x_k)$ and the smallest possible $b(x_k)$ multiplied by the smallest possible $y_{k-1}$. Likewise we bound the sequence above by adding the largest possible $a(x_k)$ and the largest possible $b(x_k)$ multiplied by the largest possible value for $y_{k-1}$ 
            \begin{equation}
                \begin{split}
                    y_0&\leq y_0\leq y_0\\
                    a^-+b^-y_0&\leq y_1\leq a^++b^+y_0\\
                    a^-+a^-b^-+(b^-)^2y_0&\leq y_2\leq a^++a^+b^++(b^+)^2y_0\\
                    a^-(1+b^-+(b^-)^2)+(b^-)^3y_0&\leq y_3\leq a^+(1+b^++(b^+)^2)+(b^+)^3y_0\\
                    &\vdots\\
                    a^-\left(\sum_{i=0}^{k-1}(b^-)^i\right) + y_0(b^-)^k&\leq y_k\leq a^+\left(\sum_{i=0}^{k-1}(b^+)^i\right) + y_0(b^+)^k
                \end{split}
            \end{equation}
            Define the following functions
            \begin{equation}
                \begin{split}
                    f_k^-(\delta) &= (a(x)-\delta)\frac{1-(b(x)-\delta)^k}{1-b(x)+\delta} + y_0(b(x)-\delta)^k\\
                    f_k^+(\delta) &= (a(x)+\delta)\frac{1-(b(x)+\delta)^k}{1-b(x)-\delta}+y_0(b(x)+\delta)^k
                \end{split}
            \end{equation}
            and notice that $f_k^-(\delta)\leq y_k\leq f_k^+(\delta)$. Moreover, notice that both $f_k^+$ and $f_k^-$ are continuous at $0$ for any choice of $k$. 
    
            Let $\varepsilon>0$ and observe that $\exists K_1$ such that if $k>K_1$ then $|f_k^-(\delta)-\frac{a(x)-\delta}{1-b(x)+\delta}|<\frac{\epsilon}{2}$ and $|f_k^+(\delta)-\frac{a(x)+\delta}{1-b(x)-\delta}|<\frac{\epsilon}{2}$. Additionally, because $f(\delta):=\frac{a+\delta}{1-b-\delta}$ is continuous at $\delta =0$, $\exists$ a $\zeta$ such that $|\delta|<\zeta \implies |\frac{a+\delta}{1-b-\delta}-\frac{a}{1-b}|<\frac{\varepsilon}{2}$. Therefore, for any $\varepsilon$ there exists a $\zeta$ and a $K_1$ such that if $k>K_1$ and $\delta<\zeta$
            \begin{equation}
                \begin{split}
                \left|f^+_k(\delta)-\frac{a(x)}{1-b(x)}\right|&\leq \left|f_k^+(\delta)-\frac{a(x)+\delta}{1-b(x)-\delta}\right|+\left|\frac{a(x)+\delta}{1-b(x)-\delta}-\frac{a(x)}{1-b(x)}\right|\\
                    &\leq \frac{\varepsilon}{2}+\frac{\varepsilon}{2}=\varepsilon\\
                \left|f^-_k(\delta)-\frac{a(x)}{1-b(x)}\right|&\leq \left|f_k^-(\delta)-\frac{a(x)-\delta}{1-b(x)+\delta}\right|+\left|\frac{a(x)-\delta}{1-b(x)+\delta}-\frac{a(x)}{1-b(x)}\right|\\
                    &\leq \frac{\varepsilon}{2}+\frac{\varepsilon}{2}=\varepsilon
                \end{split}
            \end{equation}
    
            Thus we say that 
            \begin{equation}
                \begin{split}
                    f_k^-(\delta)\leq& y_k\leq f_k^+(\delta)\\
                    f_k^-(\delta)-\frac{a(x)}{1-b(x)}\leq&y_k-\frac{a(x)}{1-b(x)}\leq f_k^+(\delta)-\frac{a(x)}{1-b(x)}\\
                    -\varepsilon\leq &y_k-\frac{a(x)}{1-b(x)}\leq \varepsilon
                \end{split}
            \end{equation}
    
            Observe that this $\zeta$ depends only of $a(x)$ and $b(x)$ so surely for any $\varepsilon$ we can find a $\delta<\zeta$ such that the above inequality is satisfied. Recall that we reindexed so that $(x_0,y_0)=(x_{K_0},y_{K_0})$. By construction, $K_0$ was such that $|b(x_k)-b(x)|<\delta<\zeta$ and $|a(x_k)-a(x)|<\delta<\zeta$ for $k>k_0$ and such a $K_0$ can be found for any value of $\zeta$ because $x_k\rightarrow x$.
            Let $K^*=K_0+K_1$ and we see that $\exists K^*$ such that $|y_k-\frac{a(x)}{1-b(x)}|<\epsilon$ when $k>K^*$. This can be done for any $\varepsilon$ so indeed $y_k\rightarrow \frac{a(x)}{1+b(x)}$. This completes the proof in the case $b(x)>0$.
    
            Suppose $b(x)=0$. For a $\delta$ which is determined later in the proof, there is a $K_0$ such that $|b(x_k)|<\delta$ and $|a(x_k)-a(x)|<\delta$. Using the same notation as before I will show that $y_k\rightarrow \frac{a(x)}{1-b(x)}=a(x)$. This time I will bound every term in the sequence by either adding the smallest $a(x_k)$ and the most negative $b(x_k)$ times the $y_{k-1}$ with the greatest magnitude or by adding the largest $a(x_k)$ and the most positive $b(x_k)$ multiplied by the largest possible magnitude of $y_{k-1}$. Again, reindex so that $(x_0,y_0)=(x_{K_0},y_{K_0})$
    
            \begin{equation}
                \begin{split}
                    y_0&\leq y_0\leq y_0\\
                    a^--\delta y_0&\leq y_1\leq a^++\delta y_0\\
                    a^--\delta(a^++\delta y_0)&\leq y_2\leq a^++\delta(a^++\delta y_0)\\
                    a^--\delta a^+(1+\delta)-\delta^3 y_0&\leq y_3\leq a^+(1+\delta +\delta^2)+\delta^3y_0\\
                    &\vdots\\
                    a^--\delta a^+\left(\sum_{i=1}^{k-2}\delta^i\right) -\delta^ky_0&\leq y_k\leq a^+\left(\sum_{i=1}^{k-1}\delta^i\right)+\delta^ky_0
                \end{split}
            \end{equation}
            as before I will define 
            \begin{equation}
                \begin{split} 
                    g_k^-(\delta)&= (a(x)-\delta)-\delta(a(x)+\delta) \frac{1-\delta^{k-1}}{1-\delta}-\delta^ky_0\\
                    g_k^+(\delta)&=(a(x)+\delta)\frac{1-\delta^{k}}{1-\delta}+\delta^ky_0    
                \end{split} 
            \end{equation}
            and again we note that continuous $\delta =0$ for any choice of $k$.
    
            As before, let $\varepsilon>0$ and observe that $\exists K_1$ such that $|g_k^-(\delta)-(a(x)-\delta)+\delta\frac{a(x)+\delta}{1-\delta}|<\frac{\varepsilon}{2}$ and $|g_k^+(\delta)-(a(x)+\delta)\frac{1}{1-\delta}|<\frac{\varepsilon}{2}$ when $k>K_1$.
            It is clear to see also that $\exists \zeta$ such that if $|\delta|<\zeta$ then $|(a(x)-\delta)+\delta\frac{a(x)+\delta}{1-\delta}-a(x)|<\frac{\varepsilon}{2}$ 
            and $|(a(x)+\delta)\frac{1}{1-\delta}|<\frac{\varepsilon}{2}$. Therefore we can say that there exists a $K_1$ and $\zeta$ such that $k>K_1$ and $\delta<\zeta$ implies that
            \begin{equation}
                \begin{split} 
                    |g_k^+(\delta)-a(x)|&\leq \left|g_k^+(\delta)-(a(x)+\delta)\frac{1}{1-\delta}\right|+\left|(a(x)+\delta)\frac{1}{1-\delta}-a(x)\right|\\
                    &\leq \frac{\varepsilon}{2}+\frac{\varepsilon}{2}=\varepsilon\\
                    |g_k^-(\delta)-a(x)|&\leq \left|g_k^-(\delta)-(a(x)-\delta)+\delta\frac{a(x)+\delta}{1-\delta}\right|+\left|(a(x)-\delta)+\delta\frac{a(x)+\delta}{1-\delta}-a(x)\right|\\
                    &\leq \frac{\varepsilon}{2}+\frac{\varepsilon}{2}=\varepsilon
                \end{split}
            \end{equation}
            Using the same squeeze theorem argument as before we can say that that for any $\varepsilon$ there is a $K_1$ and $\zeta$ such that if $k>K_1$ and $|\delta|<\zeta$ then $|y_k-a(x)|<\epsilon$.
    
            To complete the proof in this case we recall that we reindexed so that $(x_0,y_0)=(x_{k_0},y_{k_0})$ and that, because $\zeta$ depends only on $a(x)$ we can choose a $\delta<\zeta$ and always get a requisite $K_0$ so that $|a(x_k)-a(x)|<\delta$ and $|b(x_k)-0|<\delta$ when $k>K_0$.
            
            Let $K^*=K_0+K_1$  and so we have that if $k>K^*$ $|y_k-a(x)|<\varepsilon$. This can be done for any $\varepsilon$ so $y_k\rightarrow a(x)$.
    
            The most complicated version of the proof is in the case $b(x)<0$. Again $\exists K_0$ such that $|b(x_k)-b(x)|<\delta $and $|a(x_k)-a(x)|<\delta$ when $k>K_0$. Select $\delta$ so that $-1<b(x)-\delta<b(x)+\delta<0$ and so that it satisfies the condition described later on in the proof. By selecting such a $\delta$ we get a requisite $K_0$. Again, WLOG, reindex so that $(x_0,y_0)=(x_{K_0},y_{K_0})$.
    
            The complication comes in because, when we seek to bound each term, we must use a ``zig-zag" argument. To get a lower bound we need to add the smallest possible value for $a(x_k)$ and the most negative value of $b(x_k)$ multiplied by the largest magnitude $y_{k-1}$ and to get an upper bound we must add the largest possible value of $a(x_k)$ to the least negative value of $b(x_k)$ multiplied by the smallest magnitude possible for $y_{k-1}$. The process proceeds as follows:
            \begin{equation}
                \begin{split}
                    y_0\leq &y_0\leq y_0\\
                    a^-+b^-y_0\leq &y_1\leq a^++b^+y_0\\
                    a^-+b^-(a^+-b^+y_0)\leq &y_2\leq a^++b^+(a^-+b^-y_0)\\
                    a^-(1+b^+b^-)+a^+(b^-)+(b^+)(b^-)^2y_0\leq &y_3\leq a^-b^++a^+(1_b^+b^-)+(b^+)^2(b^-)y_0\\
                    &\vdots
                \end{split}
            \end{equation}
            When we continue this pattern we find that when $k$ is even 
            \begin{equation}
                \begin{split}
                    y_k&\leq \overline{h_k^e(\delta)}:=a^+b^-\left(\sum_{i=0}^{\frac{k}{2}-1}(b^+b^-)^i\right)+a^-\left(\sum_{i=0}^{\frac{k}{2}-1}(b^+b^-)^i\right)+(b^+b^-)^{\frac{k}{2}}y_0\\
                    y_k&\geq \underline{h_k^e(\delta)}:=a^+\left(\sum_{i=0}^{\frac{n}{2}-1}(b^+b^-)^i\right)+a^-b^+\left(\sum_{i=0}^{\frac{n}{2}-1}(b^+b^-)^i\right)+(b^+b^-)^{\frac{k}{2}}y_0
                \end{split}
            \end{equation}
            and when $k$ is odd
            \begin{equation}
                \begin{split} 
                    y_k&\leq \overline{h_k^o}:=a^-b^+\left(\sum_{i=0}^{\frac{k-1}{2}-1}(b^+b^-)^i\right)+ a^+\left(\sum_{i=0}^{\frac{k+1}{2}-1}(b^+b^-)^i\right)+(b^+b^-)^\frac{n-1}{2}b^+y_0\\
                    y_k&\geq \underline{h_k^o}:=a^-\left(\sum_{i=0}^{\frac{k+1}{2}-1}(b^+b^-)^i\right)+ a^+b^-\left(\sum_{i=0}^{\frac{k-1}{2}-1}(b^+b^-)^i\right)+(b^+b^-)^\frac{n-1}{2}b^+y_0
                \end{split} 
            \end{equation}
            Now observe that the even and odd subsequences of these upper and lower bounds have the same limit as $k\rightarrow \infty$.
            \begin{equation}
                \begin{split} 
                    \lim\overline{h_k^o(\delta)}&=\lim \overline{h_k^e(\delta)}=\overline{h^*(\delta)}:=a^+\left(\frac{1}{1-b^+b^-}\right) +a^-b^+\left(\frac{1}{a-b^+b^-}\right)\\
                    \lim\underline{h_k^o}&=\lim \underline{h_k^e}=\underline{h^*(\delta)}:=a^+b^-\left(\frac{1}{1-b^+b^-}\right) +a^-\left(\frac{1}{a-b^+b^-}\right)
                \end{split}
            \end{equation}
            which implies that $\overline{h_k}\rightarrow \overline{h^*(\delta)}$ and $\underline{h_k}\rightarrow \underline{h^*(\delta)}$
    
            Moreover we can see that as $\delta\rightarrow 0$
            \begin{equation}
                \begin{split} 
                    \lim_{\delta\rightarrow 0}\overline{h^*(\delta)}&=a\left(\frac{b+1}{1-b^2}\right)=\frac{a}{1-b}\\
                    \lim_{\delta\rightarrow 0}\underline{h^*(\delta)}&=a\left(\frac{1+b}{1-b^2}\right)=\frac{a}{1-b}\\
                \end{split} 
            \end{equation}
    
            This mean we can use a similar argument as in the previous two cases to say that for any $\epsilon$, $\exists K_1$ and a $\zeta$ such that if $k>K_1$ and $|\delta<\zeta$ then $|y_k-\frac{a(x)}{1-b(x)}|\leq \varepsilon$. Moreover, we reindexed so that $(x_0,y_0)=(x_{k_0},y_{k_0})$ where $K_0$ was constructed such that $|a(x_k)-a(x)|<\delta <\zeta$ and $|b(x_k)-b(x)|<\delta<\zeta$ when $k>K_0$. Thus if we let $K^*=K_0+K_1$ we have a $K$ such that $|y_k-\frac{a(x)}{1-b(x)}|<\varepsilon$ whenever $k>K$. This completes the proof in the final case and so we have the desired result. 
        \end{proof}

    \bibliography{GameTheoreticWTN}

\begin{thebibliography}{10}
\expandafter\ifx\csname url\endcsname\relax
  \def\url#1{\texttt{#1}}\fi
\expandafter\ifx\csname urlprefix\endcsname\relax\def\urlprefix{URL }\fi
\expandafter\ifx\csname href\endcsname\relax
  \def\href#1#2{#2} \def\path#1{#1}\fi

\bibitem{meeks2024wildlife}
D.~Meeks, O.~Morton, D.~P. Edwards, Wildlife farming: Balancing economic and
  conservation interests in the face of illegal wildlife trade, People and
  Nature 6~(2) (2024) 446--457.
\newblock \href {https://doi.org/10.1002/pan3.10588}
  {\path{doi:10.1002/pan3.10588}}.

\bibitem{patel2015quantitative}
N.~G. Patel, C.~Rorres, D.~O. Joly, J.~S. Brownstein, R.~Boston, M.~Z. Levy,
  G.~Smith, Quantitative methods of identifying the key nodes in the illegal
  wildlife trade network, Proceedings of the National Academy of Sciences
  112~(26) (2015) 7948--7953.
\newblock \href {https://doi.org/10.1073/pnas.1500862112}
  {\path{doi:10.1073/pnas.1500862112}}.

\bibitem{morton2021impacts}
O.~Morton, B.~R. Scheffers, T.~Haugaasen, D.~P. Edwards, Impacts of wildlife
  trade on terrestrial biodiversity, Nature Ecology \& Evolution 5~(4) (2021)
  540--548.
\newblock \href {https://doi.org/10.1038/s41559-021-01399-y}
  {\path{doi:10.1038/s41559-021-01399-y}}.

\bibitem{mitchell2024growth}
D.~V. Mitchell, S.~Woloszynek, M.~W. Mitchell, D.~T. Cronin, Z.~Zhao, G.~R.
  Rosen, M.~P. O’Connor, M.~Fero~Me{\~n}e, M.~K. Gonder, Growth and
  globalization of the central african wildlife economy: Insights from a
  23-year study of wild meat markets on bioko island, equatorial guinea, PLOS
  Sustainability and Transformation 3~(11) (2024) e0000139.
\newblock \href {https://doi.org/10.1371/journal.pstr.0000139}
  {\path{doi:10.1371/journal.pstr.0000139}}.

\bibitem{Gippert2023global}
J.~M.~W. Gippet, J.~M. Olivia K~Bates, C.~Bertelsmeier, The global risk of
  infections disease emergence from giant land snail invasion and pet trade,
  Parasites \& Vectors 16 (2023) 363.
\newblock \href {https://doi.org/10.1186/s13071-023-06000-y}
  {\path{doi:10.1186/s13071-023-06000-y}}.

\bibitem{Heinse2016Risk}
L.~M. Heinse, L.~H. Hardesty, R.~B. Harris, Risk of pathogen spillover to
  bighorn sheep from domestic sheep and goat flocks on private land, Wildlife
  Society Bulletin 40~(4) (2016) 625--633.
\newblock \href {https://doi.org/https://doi.org/10.1002/wsb.718}
  {\path{doi:https://doi.org/10.1002/wsb.718}}.

\bibitem{biggs2023governance}
D.~Biggs, A.~J. Peel, C.~Astaras, A.~Braczkowski, H.~Cheung, C.-Y. Choi, R.~D.
  Orume, H.~C{\'a}ceres-Escobar, J.~Phelps, R.~K. Plowright, et~al., Governance
  principles for the wildlife trade to reduce spillover and pandemic risk, CABI
  One Health~(2023) (2023) ohcs202300013.
\newblock \href {https://doi.org/10.1079/cabionehealth.2023.001}
  {\path{doi:10.1079/cabionehealth.2023.001}}.

\bibitem{Johnson2020Global}
C.~K. Johnson, P.~L. Hitchens, P.~S. Pandit, J.~Rushmore, T.~S. Evans, C.~C.~W.
  Young, M.~M. Doyle, Global shifts in mammalian population trends reveal key
  predictors of virus spillover risk, Proceedings of the Royal Society B:
  Biological Sciences 287~(1924) (2020) 20192736.
\newblock \href {https://doi.org/10.1098/rspb.2019.2736}
  {\path{doi:10.1098/rspb.2019.2736}}.

\bibitem{Magouras2020Emerging}
I.~Magouras, V.~J. Brookes, F.~Jori, A.~Martin, D.~U. Pfeiffer, S.~Dürr,
  Emerging zoonotic diseases: Should we rethink the animal–human interface?,
  Frontiers in Veterinary Science 7 (2020).
\newblock \href {https://doi.org/10.3389/fvets.2020.582743}
  {\path{doi:10.3389/fvets.2020.582743}}.

\bibitem{albers2020disease}
H.~J. Albers, K.~D. Lee, J.~R. Rushlow, C.~Zambrana-Torrselio, Disease risk
  from human--environment interactions: environment and development economics
  for joint conservation-health policy, Environmental and Resource Economics
  76~(4) (2020) 929--944.
\newblock \href {https://doi.org/10.1007/s10640-020-00449-6}
  {\path{doi:10.1007/s10640-020-00449-6}}.

\bibitem{horan2015managing}
R.~D. Horan, E.~P. Fenichel, D.~Finnoff, C.~A. Wolf, Managing dynamic
  epidemiological risks through trade, Journal of Economic Dynamics and Control
  53 (2015) 192--207.
\newblock \href {https://doi.org/10.1016/j.jedc.2015.02.005}
  {\path{doi:10.1016/j.jedc.2015.02.005}}.

\bibitem{chu2024game}
W.~Chu, Y.~Shi, X.~Jiang, T.~Ciano, B.~Zhao, Game theory approach for secured
  supply chain management in effective trade management, Annals of Operations
  Research (2024) 1--19\href {https://doi.org/10.1007/s10479-023-05792-7}
  {\path{doi:10.1007/s10479-023-05792-7}}.

\bibitem{zhang2013farsighted}
J.~Zhang, L.~Xue, L.~Zu, Farsighted free trade networks, International Journal
  of Game Theory 42 (2013) 375--398.
\newblock \href {https://doi.org/10.1007/s00182-013-0366-x}
  {\path{doi:10.1007/s00182-013-0366-x}}.

\bibitem{zhao2021complex}
C.~Zhao, X.~Xie, J.~Song, Complex network game model simulation of arctic
  sustainable fishery trade cooperation under covid-19, Sustainability 13~(14)
  (2021) 7626.
\newblock \href {https://doi.org/10.3390/su13147626}
  {\path{doi:10.3390/su13147626}}.

\bibitem{hill2010infectious}
A.~L. Hill, D.~G. Rand, M.~A. Nowak, N.~A. Christakis, Infectious disease
  modeling of social contagion in networks, PLOS computational biology 6~(11)
  (2010) e1000968.
\newblock \href {https://doi.org/10.1371/journal.pcbi.1000968}
  {\path{doi:10.1371/journal.pcbi.1000968}}.

\bibitem{liu2014controlling}
S.~Liu, N.~Perra, M.~Karsai, A.~Vespignani, Controlling contagion processes in
  activity driven networks, Physical review letters 112~(11) (2014) 118702.
\newblock \href {https://doi.org/10.1103/PhysRevLett.112.118702}
  {\path{doi:10.1103/PhysRevLett.112.118702}}.

\bibitem{taylor2015topological}
D.~Taylor, F.~Klimm, H.~A. Harrington, M.~Kram{\'a}r, K.~Mischaikow, M.~A.
  Porter, P.~J. Mucha, Topological data analysis of contagion maps for
  examining spreading processes on networks, Nature communications 6~(1) (2015)
  7723.
\newblock \href {https://doi.org/10.1038/ncomms8723}
  {\path{doi:10.1038/ncomms8723}}.

\bibitem{fefferman2007disease}
N.~Fefferman, K.~Ng, How disease models in static networks can fail to
  approximate disease in dynamic networks, Physical Review E—Statistical,
  Nonlinear, and Soft Matter Physics 76~(3) (2007) 031919.
\newblock \href {https://doi.org/10.1103/PhysRevE.76.031919}
  {\path{doi:10.1103/PhysRevE.76.031919}}.

\bibitem{hinz2011social}
O.~Hinz, B.~Skiera, C.~Barrot, J.~U. Becker, Social contagion--an empirical
  comparison of seeding strategies for viral marketing, Journal of Marketing
  75~(6) (2011) 55--71.
\newblock \href {https://doi.org/10.1509/jm.10.0088}
  {\path{doi:10.1509/jm.10.0088}}.

\bibitem{seiler2013strategic}
M.~Seiler, A.~Collins, N.~Fefferman, Strategic mortgage default in the context
  of a social network: an epidemiological approach, Journal of Real Estate
  Research 35~(4) (2013) 445--476.
\newblock \href {https://doi.org/10.1080/10835547.2013.12091371}
  {\path{doi:10.1080/10835547.2013.12091371}}.

\bibitem{cristancho2021accounting}
L.~Cristancho~Fajardo, P.~Ezanno, E.~Vergu, Accounting for farmers’ control
  decisions in a model of pathogen spread through animal trade, Scientific
  reports 11~(1) (2021) 9581.
\newblock \href {https://doi.org/10.1038/s41598-021-88471-6}
  {\path{doi:10.1038/s41598-021-88471-6}}.

\bibitem{palisson2016role}
A.~Palisson, A.~Courcoul, B.~Durand, Role of cattle movements in bovine
  tuberculosis spread in france between 2005 and 2014, PLoS One 11~(3) (2016)
  e0152578.
\newblock \href {https://doi.org/10.1371/journal.pone.0152578}
  {\path{doi:10.1371/journal.pone.0152578}}.

\bibitem{Housman2023Game}
D.~Housman, R.~A. Gillman, Game Theory a Modeling approach, Chapman and Hall,
  2023.

\bibitem{McAlisterSpilloverModel}
J.~S. McAlister, Spillover model in wildlife trade (2025).
\newblock \href {https://doi.org/https://doi.org/10.5281/zenodo.15121488}
  {\path{doi:https://doi.org/10.5281/zenodo.15121488}}.

\bibitem{cavasos2023understanding}
K.~Cavasos, R.~K. Adhikari, N.~C. Poudyal, J.~L. Brunner, A.~R. Warwick, M.~J.
  Gray, Understanding the demand for and value of pathogen-free amphibians to
  us pet owners, Conservation Science and Practice 5~(9) (2023) e12995.
\newblock \href {https://doi.org/10.1111/csp2.12995}
  {\path{doi:10.1111/csp2.12995}}.

\bibitem{cavasos2023attitudes}
K.~Cavasos, N.~C. Poudyal, J.~L. Brunner, A.~R. Warwick, J.~Jones, N.~Moherman,
  M.~George, J.~D. Willard, Z.~T. Brinks, M.~J. Gray, Attitudes and behavioral
  intentions of pet amphibian owners about biosecurity practices, EcoHealth
  20~(2) (2023) 194--207.
\newblock \href {https://doi.org/10.1007/s10393-023-01645-8}
  {\path{doi:10.1007/s10393-023-01645-8}}.

\bibitem{smith2017summarizing}
K.~M. Smith, C.~Zambrana-Torrelio, A.~White, M.~Asmussen, C.~Machalaba,
  S.~Kennedy, K.~Lopez, T.~M. Wolf, P.~Daszak, D.~Travis, et~al., Summarizing
  us wildlife trade with an eye toward assessing the risk of infectious disease
  introduction, EcoHealth 14 (2017) 29--39.
\newblock \href {https://doi.org/10.1007/s10393-017-1211-7}
  {\path{doi:10.1007/s10393-017-1211-7}}.

\bibitem{can2019dealing}
{\"O}.~E. Can, N.~D'Cruze, D.~W. Macdonald, Dealing in deadly pathogens: Taking
  stock of the legal trade in live wildlife and potential risks to human
  health, Global Ecology and conservation 17 (2019) e00515.
\newblock \href {https://doi.org/10.1016/j.gecco.2018.e00515}
  {\path{doi:10.1016/j.gecco.2018.e00515}}.

\bibitem{berg2020compliance}
C.~Berg, H.~Frida~Lundmark, Compliance with animal welfare regulations: drivers
  and consequences, CABI Reviews~(2020) (2020).
\newblock \href {https://doi.org/10.1079/PAVSNNR202015025}
  {\path{doi:10.1079/PAVSNNR202015025}}.

\end{thebibliography}
\end{document}